\documentclass[sigconf,nonacm]{acmart} 
\setcopyright{none}
\renewcommand\footnotetextcopyrightpermission[1]{}
\AtBeginDocument{%
  }

\usepackage{bbm}
\usepackage{enumitem}
\usepackage{adjustbox}
\newcommand{\myparagraph}[1]{\vspace{0.5mm} \noindent\textbf{#1}}
\usepackage{multirow}
\usepackage{booktabs}
\usepackage{graphicx}
\usepackage{array}
\newcolumntype{C}[1]{>{\centering\arraybackslash}m{#1}}
\newcolumntype{L}[1]{>{\raggedright\arraybackslash}m{#1}}
\usepackage{subcaption}
\usepackage{colortbl}
\usepackage[table]{xcolor} 
\newcommand{\down}{\ensuremath{\downarrow}}
\newcommand{\up}{\ensuremath{\uparrow}}
\usepackage{pifont}  
\newcommand{\xmark}{\ding{55}}
\newtheorem{definition}{Definition}
\usepackage{amsthm}
\newtheorem{theorem}{Theorem}
\newtheorem{lemma}{Lemma}
\usepackage[ruled,vlined]{algorithm2e}

\begin{document}

\title{Universal Graph Backdoor Defense: A Feature-based Homophily Perspective}

\author{Mengting Pan}
\email{mengting.pan@unsw.edu.au}
\affiliation{%
  \institution{The University of New South Wales}
  \city{Sydney}
  \country{Australia}
}

\author{Fan Li}
\email{fan.li8@unsw.edu.au}
\affiliation{%
  \institution{The University of New South Wales}
  \city{Sydney}
  \country{Australia}
}

\author{Chen Chen}
\email{chenc@uow.edu.au}
\affiliation{%
  \institution{University of Wollongong}
  \city{Wollongong}
  \country{Australia}
}

\author{Xiaoyang Wang}
\email{xiaoyang.wang1@unsw.edu.au}
\affiliation{%
  \institution{The University of New South Wales}
  \city{Sydney}
  \country{Australia}
}
\renewcommand{\shortauthors}{Pan et al.}

\begin{abstract}
Graph neural networks (GNNs) have achieved remarkable success in relational learning. However, their vulnerability to graph backdoor attacks (GBAs) poses a significant barrier to broader adoption in high-stakes applications.
Despite recent advances in graph backdoor defense (GBD), existing methods primarily focus on subgraph-based GBAs, relying on the assumption that poisoned target nodes are explicitly connected to subgraph triggers. Our empirical results reveal that such structure-centric approaches fail to defend against emerging feature-based GBAs that preserve graph topology. Therefore, in this paper, we study a novel problem of universal graph backdoor defense. First, we investigate the shared effects of both attack types from a feature-based homophily perspective, which characterizes local feature consistency between nodes and their neighborhoods.
Thorough theoretical and empirical analyses demonstrate that, regardless of trigger mechanisms, backdoors induced by GBAs exhibit lower feature-based homophily than clean nodes, indicating a discrepancy in local feature similarity.
Motivated by this insight, we propose to leverage node-level local feature consistency, modeled by a neighbor-aware reconstruction loss, to distinguish backdoors from clean nodes.
Then, a robust training strategy is developed
to eliminate trigger effects while reducing noise induced by detection uncertainty.
Extensive experiments demonstrate that our framework significantly degrades the attack success rate and maintains competitive clean accuracy under both subgraph-based and feature-based attacks.

\end{abstract}





\maketitle

\section{Introduction}

Graph-structured data~\cite{graph2018,graph2020,graph2020-2} is widely used for relational modeling across many real-world domains, including social networks~\cite{fan2019graph,social2022}, molecular biology~\cite{molecular2021,quesado2024hybrid}, and financial systems~\cite{financial2021,li2024adarisk}.
To effectively learn from such data, Graph Neural Networks (GNNs)~\cite{wu2020comprehensive,liu2023beyond,li2025tcgu} have emerged as a prevalent and powerful model, leveraging message passing to encode local structural patterns and node features. Despite their remarkable success, recent studies~\cite{ugba2023,dpgba2024,spear2025} reveal that GNNs are vulnerable to backdoor attacks. 
In these attacks, adversaries inject malicious triggers into selected target nodes and manipulate their labels toward a predefined target class, causing the trained model to learn the spurious trigger–class association and produce targeted misclassification during inference.
This vulnerability hinders the safe deployment of GNNs in high-stakes scenarios such as banking systems, cybersecurity, and healthcare.

\begin{figure}[t]
  \centering
  \begin{subfigure}[b]{0.48\linewidth}
    \centering
    \includegraphics[width=\linewidth]{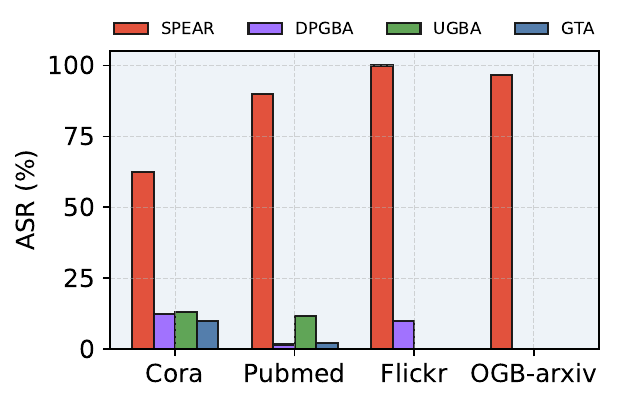}
    \caption{The performance of RIGBD.}
    \vspace{-1em}
    \label{fig:introduction}
  \end{subfigure}
  \hfill
  \begin{subfigure}[b]{0.48\linewidth}
    \centering
    \includegraphics[width=\linewidth]{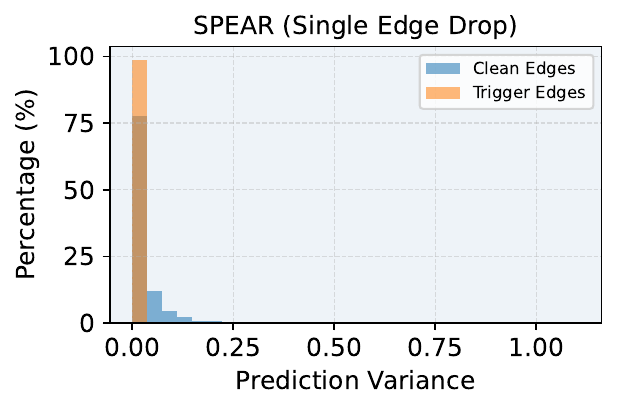}
    \caption{SPEAR (Single Edge Drop).}
     \vspace{-1em}
    \label{fig:spear-variance}
  \end{subfigure}
  \caption{ (a): Attack success rate (\%) of different graph backdoor attacks under RIGBD. (b): Prediction variance caused by dropping trigger edges and clean edges. }
\Description{}
  \vspace{-1em}
  \label{fig:rigbd_analysis}
\end{figure}

Existing graph backdoor attacks primarily focus on designing subgraph-based triggers to manipulate local neighborhood patterns around poisoned nodes.
For example, SBA~\cite{SBA2021} adopts random or sampled subgraphs as triggers, while GTA~\cite{gta2021} introduces learnable triggers that adapt to individual targets.
To further enhance stealthiness, UGBA~\cite{ugba2023} and DPGBA~\cite{dpgba2024} incorporate regularization mechanisms to generate similarity-preserving and in-distribution triggers, respectively, making trigger nodes less distinguishable from clean ones. To defend against the above subgraph-based attacks, several initial methods, such as Prune~\cite{ugba2023} and OD~\cite{dpgba2024}, have been proposed.
However, these approaches are tailored to specific backdoor patterns and thus remain effective only for particular attack methods.
To address this limitation, RIGBD~\cite{rigbd2024} proposes to identify poisoned target nodes through prediction variance under random edge perturbations and further introduces a robust training strategy to eliminate their adverse effects.
This enables a more effective and generalizable defense against subgraph-based attacks.

Despite achieving state-of-the-art performance against backdoor attacks based on subgraph triggers, RIGBD fails on SPEAR~\cite{spear2025}, a recent feature-based attack that implants triggers into node attributes without altering graph topology, as shown in Fig.~\ref{fig:rigbd_analysis}(a).
Specifically, we observe that although RIGBD consistently reduces the attack success rate (ASR) of subgraph-based attacks to below 15\% across datasets, SPEAR still achieves over 60\% ASR under RIGBD's defense, even exceeding 89\% on most datasets. 
This failure arises because RIGBD relies on the empirical assumption that removing edges associated with backdoor triggers induces a large prediction variance on poisoned target nodes.
In contrast, feature-based attacks do not introduce explicit structural triggers, which invalidates this assumption.
To substantiate this claim, we iteratively drop the neighboring edges around each node and measure the resulting variance in predictions.
Our empirical results in Fig.~\ref{fig:rigbd_analysis}(b) show that under SPEAR, edge dropping around poisoned nodes induces prediction variance indistinguishable from that of clean nodes, further explaining the ineffectiveness of RIGBD against feature-based attacks.

To bridge this research gap, we first investigate the shared effects in subgraph-based and feature-based backdoor attacks.
Intuitively, regardless of whether triggers are attached to poisoned targets in the form of subgraphs or directly injected into node attributes, the feature similarity between poisoned targets (or trigger nodes) and their surrounding neighbors may be perturbed. Motivated by this, we introduce a novel concept termed \textit{feature-based homophily}, which characterizes the feature consistency between each node and its local neighborhood. Through rigorous theoretical analysis in Sec.~\ref{sec:rethinking}, we show that, compared with clean nodes, poisoned targets exhibit lower feature-based homophily with respect to their surrounding neighbors, indicating stronger feature inconsistency within local neighborhoods. Empirical results in Sec.~\ref{sec:empirical-analysis} across multiple datasets further strengthen this theoretical finding.
This homophily discrepancy unifies the intrinsic effects of both attack types and naturally provides an effective discriminative signal for identifying backdoors, including both triggers and poisoned targets.

Based on these insights, we propose CoGBD (\textbf{Co}nsistency-guided \textbf{G}raph \textbf{B}ackdoor \textbf{D}efense), a novel graph backdoor defense framework that 
leverages the deviation of feature consistency within local neighborhoods.
To the best of our knowledge, CoGBD is the first universal graph backdoor defense capable of effectively defending against both subgraph-based and feature-based attacks. Overall, CoGBD adopts a two-stage pipeline.
Motivated by the above analysis, in the first stage, we design tri-directional reconstruction objectives to model feature consistency within the local neighborhood of each individual node, where nodes exhibiting large reconstruction errors are identified as potential backdoors. 
In the second stage, we develop a noise-aware robust training strategy that minimizes over-confident predictions for candidate poisoned nodes using adaptive weights inversely related to their reconstruction errors.
This design effectively counteracts the impacts of triggers while mitigating noise introduced by potentially imprecise backdoor detection.
Our main contributions are summarized as follows:
\begin{itemize}

    \item We empirically show the ineffectiveness of state-of-the-art graph backdoor defenses against feature-based attacks and study a novel problem of universal graph backdoor defense.
    
    \item We unify the effects of both subgraph-based and feature-based graph backdoor attacks from a feature-based homophily perspective. Our theoretical and empirical analyses uncover a local feature consistency gap that provides practical guidance for identifying backdoors.
    \item We propose CoGBD, the first universal graph backdoor defense framework, consisting
    (i) a reconstruction-driven backdoor detection stage that leverages the feature consistency gap 
    and (ii) a noise-aware robust training stage that adaptively counteracts trigger influence.
    \item Extensive experiments demonstrate the effectiveness of CoGBD in defending against both subgraph-based and feature-based attacks while maintaining clean accuracy.
\end{itemize}

\section{Related Works}
\label{sec:related}

\myparagraph{Graph Backdoor Attacks.}
Graph backdoor attacks (GBAs)~\cite{poster2022,mlgba2024,eumc2024,cgba2025} expose a severe vulnerability of GNNs, where an adversary implants triggers into a small set of training nodes so that the trained model associates the trigger with a target label.
At test time, any node carrying the trigger is misclassified to the target class, while predictions on clean nodes remain largely unaffected.
Most existing GBAs instantiate triggers as subgraphs by manipulating local connectivity patterns.
SBA~\cite{SBA2021} injects a universal subgraph into selected nodes, and GTA~\cite{gta2021} further learns adaptive triggers via a generator.
To improve stealthiness, UGBA~\cite{ugba2023} encourages similarity between trigger nodes and the attached targets.
Subsequent studies show that many of these triggers are out-of-distribution and remain distinguishable from clean neighborhoods~\cite{dpgba2024}.
DPGBA~\cite{dpgba2024} improves stealthiness by generating in-distribution triggers through adversarial optimization.
Beyond structural manipulation, feature-based graph backdoor attacks inject triggers
directly into node attributes while keeping the graph topology intact.
SPEAR~\cite{spear2025} exemplifies this paradigm by embedding attack signals purely in
node features, significantly increasing stealthiness and challenging defenses that
rely on structural irregularities.


\myparagraph{Graph Backdoor Defense.}
Compared with the rapidly growing studies on graph backdoor attacks, graph backdoor defense (GBD) methods remain limited. 
Most existing GBD methods are designed for subgraph-based attack and rely on the assumption that backdoor triggers introduce explicit structural irregularities.
For example, Prune~\cite{ugba2023} weakens potential trigger effects by removing edges between node pairs with low feature similarity, while Outlier Detection
(OD)~\cite{dpgba2024} identifies trigger nodes that deviate from the distribution of clean data and isolates their connections from the remaining graph.
However, these methods are typically tailored to specific trigger structures and are effective only for particular subgraph-based attack patterns.
To overcome this limitation, RIGBD~\cite{rigbd2024} focuses on poisoned target nodes by exploiting the model's sensitivity to random edge dropping and incorporating robust training to eliminate trigger effects, enabling more generalizable defense in various subgraph-based backdoor attacks.
While RIGBD achieves strong performance against a wide range of subgraph-based backdoor attacks, it becomes less reliable in recent feature-based GBAs, where backdoor triggers are embedded in node attributes and the graph topology is largely preserved.
This gap motivates us to develop a unified defense framework that remains effective against both subgraph-based and feature-based backdoor attacks.

\section{Preliminaries}

\myparagraph{Notations.}
Let
$\mathcal{G}=(\mathcal{V},\mathcal{E},\mathbf{X})$ denote an attributed graph,
where $\mathcal{V}=\{v_1,\ldots,v_N\}$ is the set of $N$ nodes,
$\mathcal{E}\subseteq\mathcal{V}\times\mathcal{V}$ is the set of $E$ edges,
and $\mathbf{X}\in\mathbb{R}^{N\times F}$ is the node attribute matrix,
with each node $v_i$ associated with a feature vector $\mathbf{x}_i\in\mathbb{R}^F$. The one-hop neighborhood of node $v$ is denoted by $\mathcal{N}(v)=\{u:(u,v)\in\mathcal{E}\}$.
$\mathbf{A}\in\mathbb{R}^{N\times N}$ is the adjacency matrix of the graph $\mathcal{G}$, where $A_{ij}=1$ if $(v_i,v_j)\in\mathcal{E}$, and $A_{ij}=0$ otherwise.
We denote the normalized adjacency matrix with self-loops as
$\bar{\mathbf{A}}=\mathbf{D}^{-1/2}(\mathbf{A}+\mathbf{I})\mathbf{D}^{-1/2}$,
where $\mathbf{I}$ is the identity matrix and $\mathbf{D}$ is the diagonal degree matrix
with $\mathbf{D}_{ii}=\sum_j(\mathbf{A}+\mathbf{I})_{ij}$.
In this paper, we focus on the semi-supervised node classification task in the inductive setting.
During training, we are given a graph
$\mathcal{G}_{\mathrm{T}}=(\mathcal{V}_{\mathrm{T}},\mathcal{E}_{\mathrm{T}},\mathbf{X}_{\mathrm{T}})$.
We use $\mathcal{V}_{\mathrm{C}}\subseteq\mathcal{V}_{\mathrm{T}}$ and
$\mathcal{V}_{\mathrm{B}}\subseteq\mathcal{V}_{\mathrm{T}}$
to denote the clean and backdoored node sets, respectively.
Nodes in $\mathcal{V}_{\mathrm{C}}$ are labeled with clean labels,
whereas nodes in $\mathcal{V}_{\mathrm{B}}$ are labeled with the target label $y_{\mathrm{t}}$.
The remaining nodes
$\mathcal{V}_{\mathrm{T}}\setminus(\mathcal{V}_{\mathrm{C}}\cup\mathcal{V}_{\mathrm{B}})$
are unlabeled.
At inference time, we are given an unseen graph
$\mathcal{G}_{\mathrm{U}}=(\mathcal{V}_{\mathrm{U}},\mathcal{E}_{\mathrm{U}},\mathbf{X}_{\mathrm{U}})$.
The node set $\mathcal{V}_{\mathrm{U}}$ consists of clean and backdoored test nodes,
denoted by $\mathcal{V}_{\mathrm{UC}}$ and $\mathcal{V}_{\mathrm{UB}}$, respectively.
Notably, $\mathcal{G}_{\mathrm{U}}$ is disjoint from the training graph, i.e.,
$\mathcal{V}_{\mathrm{U}}\cap\mathcal{V}_{\mathrm{T}}=\emptyset$.

\myparagraph{Graph Neural Networks (GNNs).}
Most existing GNNs adopt a message-passing paradigm, in which each node
$v \in \mathcal{V}$ is associated with a representation vector $h_v$ that is iteratively
updated over $K$ layers.
At each layer, node representations are first aggregated from local neighborhoods $\mathcal{N}(v)$ and then transformed by a learnable encoder.
Formally,  for common GNN models such as GCN~\cite{gcn2017} and GAT~\cite{gat2018}, 
the representation of node $v$ at the $k$-th layer is updated as:
$
    h_v^{(k)}=
    \mathrm{ENC}\!\left(
    \mathrm{AGGR}\big(\{h_u^{(k-1)}:u\in\mathcal{N}(v)\cup\{v\}\}\big)
    \right),
$
where $\mathrm{AGGR}(\cdot)$ denotes a permutation-invariant neighborhood aggregation
operator and $\mathrm{ENC}(\cdot)$ is a learnable transformation.

\myparagraph{Threat Model.}
We consider the common gray-box backdoor attack~\cite{gta2021,ugba2023,dpgba2024,rigbd2024,spear2025},
where the attacker has access to the training graph $\mathcal{G}_{\mathrm{T}}$, but has no knowledge of the victim GNN architecture or parameters.
The attacker aims to attach backdoor triggers, i.e., nodes, subgraphs, or feature perturbations, into a small set of nodes $\mathcal{V}_{\mathrm{B}}$ and assign them the target label $y_t$,
resulting in a backdoored graph $\mathcal{G}_{\mathrm{B}}=(\mathcal{V}'_{\mathrm{T}},\mathcal{E}'_{\mathrm{T}},\mathbf{X}'_{\mathrm{T}})$.
A GNN model trained on $\mathcal{G}_{\mathrm{B}}$ is expected to behave normally on clean nodes while being misguided by the backdoor trigger to classify nodes attached with triggers as $y_{t}$.
Formally, given a surrogate GNN model $f_s$, which the attacker uses as a proxy to design and optimize backdoor triggers, the attack objective can be expressed as:
$
\min_{f_s}\;
\sum_{v \in \mathcal{V}_{\mathrm B}}
\ell\!\left(f_s(v), y_t\right)
\;+\;
\sum_{u \in \mathcal{V}_{\mathrm C}}
\ell\!\left(f_s(u), y_u\right),
$
where $\ell(\cdot,\cdot)$ denotes the classification loss.

\myparagraph{Defender's Knowledge and Objective.}
Unlike prior graph backdoor defenses~\cite{ugba2023,dpgba2024,rigbd2024} that are tailored to subgraph-based attacks and do not generalize to feature-based triggers, we study a universal, attack-agnostic defense setting.
The defender only has access to the backdoored training graph $\mathcal{G}_{\mathrm{B}}$,
without prior knowledge of the attack type, the poisoned node set $\mathcal{V}_{\mathrm{B}}$,
or the target label $y_t$.
Given a backdoored graph $\mathcal{G}_{\mathrm{B}}$, our objective is to train a GNN model $f_g$ that can defend against backdoor triggers during inference while preserving accuracy on clean data, which is formally defined as:
$
\min_{f_g}\;
\sum_{v \in \mathcal{V}_{\mathrm{UC}}}
\ell\!\left(f_g(v), y_v\right)
\;-\;
\sum_{u \in \mathcal{V}_{\mathrm{UB}}}
\ell\!\left(f_g(u), y_t\right),
$
where $\ell(\cdot,\cdot)$ denotes the classification loss.

\section{A Unified View of Graph Backdoor Attacks}
\label{sec:analysis}

In this section, we provide a unified view of graph backdoor attacks.
By rethinking both subgraph-based and feature-based GBAs
from a novel feature-based homophily perspective,
we theoretically and empirically show that
seemingly different attack strategies
induce common intrinsic effects at the representation level.
This unified perspective decouples the analysis from specific trigger designs,
exposes the fundamental mechanisms shared across GBAs,
and facilitates the design of the universal backdoor defense strategy.

\subsection{Theoretical Analysis}
\label{sec:rethinking}
\myparagraph{Feature-based Homophily.}
We posit that both subgraph-based and feature-based GBAs,
despite employing different trigger forms,
can implicitly alter the consistency between a node and its local neighborhood.
Specifically, subgraph triggers and malicious feature perturbations
may disrupt the alignment between node attributes and surrounding structures,
even when the overall graph topology is preserved.
To characterize this consistency,
we introduce a novel \textit{feature-based homophily} metric:

\begin{definition}[Feature-based Homophily]
\label{def:feat-homo}
The feature-based homophily of a node $v$ is defined as the similarity
between its features
and the aggregated features of its neighboring nodes:
\begin{equation}
\label{eq:feat-homo}
\mathcal{H}_{\mathrm{feat}}(v)
=
\mathrm{sim}\!\left(
x_v,\;
\mathrm{AGGR}\{x_u : u \in \mathcal{N}(v)\}
\right),
\end{equation}
where $\mathrm{sim}(\cdot,\cdot)$ denotes the similarity metric (e.g., cosine similarity).
\end{definition}

\myparagraph{How Effective GBAs Affect Message Passing in GNNs.} We begin by analyzing how graph backdoor attacks (GBAs) affect
message passing and the resulting node representations in GNNs.
To this end, consider a general backdoored graph
$
\mathcal{G}_{\mathrm{B}}
=
(\mathbf{A}+\Delta\mathbf{A},\; \mathbf{X}+\Delta\mathbf{X},\; \tilde{\mathbf{Y}}),
$
where $\Delta\mathbf{A}$ and $\Delta\mathbf{X}$ denote the induced perturbations
on structure and features, respectively,
and $\tilde{\mathbf{Y}}$ denotes the poisoned label set,
with target nodes assigned the target class $y_t$.
Let $\mathbf{H}^{(l)}=\bar{\mathbf{A}}^{\,l}\mathbf{X}$ denote the node representations after $L$ layers of message passing,
and let $h^{(l)}_v \triangleq [\mathbf{H}^{(l)}]_v$ be the representation of node $v$. The following lemma characterizes the shift
in hidden representations after message passing under attacks.

\begin{lemma}
\label{lem:propagation}
Let $h_v^{(l)}$ and ${h'}_v^{(l)}$ denote the $l$-th layer representations
of a target node $v$ on the clean and backdoored graphs
$\mathcal{G}_{\mathrm{T}}$ and $\mathcal{G}_{\mathrm{B}}$, respectively.
We further denote
$
\pi^{(l)}_{vu} \triangleq (\bar{\mathbf A}^{\,l})_{vu},
\;
{\pi'}^{(l)}_{vu} \triangleq
\big((\bar{\mathbf A}+\Delta\bar{\mathbf A})^{\,l}\big)_{vu}.
$
Then, the representation shift
$\Delta h_v^{(l)} \triangleq {h'}_v^{(l)} - h_v^{(l)}$
can be written as:
\begin{equation}
\label{eq:dhv}
\Delta h_v^{(l)}
=
\sum_{u\in\mathcal{V}}
\big({\pi'}^{(l)}_{vu}-\pi^{(l)}_{vu}\big)\,x_u
+
\sum_{u\in\mathcal{V}}
{\pi'}^{(l)}_{vu}\,\Delta x_u,
\end{equation}
where $x_u$ is the original feature vector of node $u$ and $\Delta x_u$ denotes the feature perturbation applied on node $u$.
\end{lemma}

\begin{proof}
The proof is provided in Appendix~\ref{App:propagation}.
\end{proof}


We next provide the necessary condition under which the above perturbations cause targeted nodes to be misclassified.

\begin{lemma}
\label{lem:target}
Consider a targeted graph backdoor attack on node $v$ with target class $y_t$.
The GNN classifier produces logits
$
z_{v,c} = \mathbf w_c^\top h_v^{(l)},
$
where $\mathbf w_c$ is the class-specific weight vector.
The predicted probability for class $c$ is
$
p_{v,c} = \frac{\exp(z_{v,c})}{\sum_{c'} \exp(z_{v,c'})}.
$
Let $\ell_{\mathrm{CE}}(v) = -\log p_{v,y_t}$ be the target-class cross-entropy loss.
For a small representation shift $\Delta h_v^{(l)}$ induced by the trigger,
a necessary condition for decreasing $\ell_{\mathrm{CE}}(v)$ is:
\begin{equation}
\label{eq:target}
\left\langle
\mathbf w_{y_t} - \bar{\mathbf w}_v,\,
\Delta h_v^{(l)}
\right\rangle > 0,
\end{equation}
where
$
\bar{\mathbf w}_v \triangleq \sum_c p_{v,c}\,\mathbf w_c
$
denotes the expected classifier weight under the current predictive distribution of node $v$.
\end{lemma}

\begin{proof}
The proof is provided in Appendix~\ref{App:shift}.
\end{proof}

\myparagraph{Remark.}
Lemma~\ref{lem:propagation} characterizes how different types of backdoor triggers perturb node representations through message passing.
Lemma~\ref{lem:target} further reveals that such perturbations are not arbitrary: the induced representation shift must be explicitly aligned with the target class direction of the classifier.

\myparagraph{Distributional Shift in Feature-Based Homophily.}
We now investigate how effective GBAs induce distributional discrepancies in feature-based homophily across node types.

\begin{theorem}
\label{the:homo-drop}
Given a backdoored graph $\mathcal G_{\mathrm B}$ with clean node set
$\mathcal V_{\mathrm C}$ and backdoored node set $\mathcal V_{\mathrm B}$.
For any effective graph backdoor attack, the expected feature-based homophily of backdoored nodes is lower than that of clean nodes:
\begin{equation}
\mathbb E_{v\sim \mathcal V_{\mathrm B}}
\!\left[\mathcal H_{\mathrm{feat}}(v)\right]
\;<\;
\mathbb E_{v\sim \mathcal V_{\mathrm C}}
\!\left[\mathcal H_{\mathrm{feat}}(v)\right].
\end{equation}
\end{theorem}

\begin{proof}
The complete proof is provided in Appendix~\ref{App:homo-drop}.
\end{proof}

\noindent\textbf{Remark.}
Theorem~\ref{the:homo-drop} reveals a distributional discrepancy in the expected feature-based homophily between clean and poisoned target nodes under effective graph backdoor attacks. This phenomenon arises consistently across both subgraph-based and feature-based attack paradigms,  highlighting a shared effect of GBAs.
In particular, subgraph-based attacks manipulate propagation paths to amplify target-aligned signals, indirectly inducing feature–neighborhood mismatch, whereas feature-based attacks directly inject target-aligned attributes that disrupt local consistency. 
Beyond target nodes, trigger nodes constitute another essential source of backdoor behavior.
Their local context 
is typically designed to maximize attack effectiveness rather than to preserve
node--neighborhood consistency. As a result, trigger nodes likewise exhibit pronounced neighborhood inconsistency as well.
Further discussion on trigger nodes is provided in Appendix~\ref{app:trigger-homo}.

\begin{table}[t]
\centering
\caption{Feature-based homophily across GBAs.}
\label{tab:label-homo}
\footnotesize
\begin{tabular}{%
C{0.9cm}L{1.6cm}
C{1.0cm}C{1.0cm}C{1.0cm}C{1.0cm}
}
\toprule
\textbf{Attack} & \textbf{Node Types} & \textbf{Cora} & \textbf{Pubmed} & \textbf{Flickr} & \textbf{OGB-arxiv} \\
\midrule

\multirow{3}{*}{\textbf{GTA}}
 & Clean Nodes   & 0.1767 & 0.2647 & 0.3201 & 0.8346 \\
 & Target Nodes  & 0.1332 & 0.1439 & 0.2273 & 0.5847 \\
 & Trigger Nodes & -0.1166 & -0.0461  &  -0.1060 & -0.1012 \\
\midrule

\multirow{3}{*}{\textbf{UGBA}}
 & Clean Nodes    & 0.1767    & 0.2647 & 0.3201 & 0.8346 \\
 & Target Nodes   & 0.1253    & 0.1936 & 0.2646 & 0.7095 \\
 & Trigger Nodes  & 0.1578    & 0.2019 & 0.2774 & 0.8001 \\
\midrule

\multirow{3}{*}{\textbf{DPGBA}}
 & Clean Nodes    & 0.1767 & 0.2647 & 0.3201  & 0.8346 \\
 & Target Nodes   & 0.1422 & 0.1777 & 0.2716  & 0.5251 \\
 & Trigger Nodes  & 0.0691 & 0.0039 & -0.1033 & 0.0914 \\
\midrule

\multirow{2}{*}{\textbf{SPEAR}}
 & Clean Nodes   & 0.1763 & 0.2642 & 0.3201 & 0.8143 \\
 & Target Nodes  & 0.1208 & 0.0887 & 0.3098 & 0.0852 \\
\bottomrule
\end{tabular}

\end{table}
\subsection{Empirical Analysis}
\label{sec:empirical-analysis}
To further validate the above theoretical finding, we empirically evaluate feature-based homophily under four representative GBAs spanning both subgraph-based and feature-based attacks. 
In particular, we report the average feature-based homophily of clean nodes, poisoned target nodes,  and trigger nodes on the backdoored graph, computed via cosine similarity.
From Table~\ref{tab:label-homo}, we observe that both poisoned target nodes and trigger nodes 
consistently exhibit substantially lower feature-based homophily than clean nodes, 
reflecting a clear homophily discrepancy between backdoors and clean nodes.
For poisoned target nodes, this gap is particularly pronounced under feature-based attacks such as SPEAR, where attribute-level triggers directly disrupt local feature--neighborhood alignment (e.g., on OGB-arxiv, target nodes show an approximately 89.5\% lower homophily value compared to clean nodes).
For trigger nodes, large homophily gaps are observed under subgraph-based attacks such as GTA and DPGBA (e.g., on OGB-arxiv, trigger nodes exhibit about 80\% and 90\% lower homophily relative to clean nodes, respectively).
These results empirically support Theorem~\ref{the:homo-drop} and demonstrate that feature-based homophily serves as a stable and discriminative signal for identifying diverse backdoor-induced anomalies across different attack paradigms.

\section{Methodology}
\begin{figure}[!t]
\centering
\includegraphics[width=1\columnwidth]{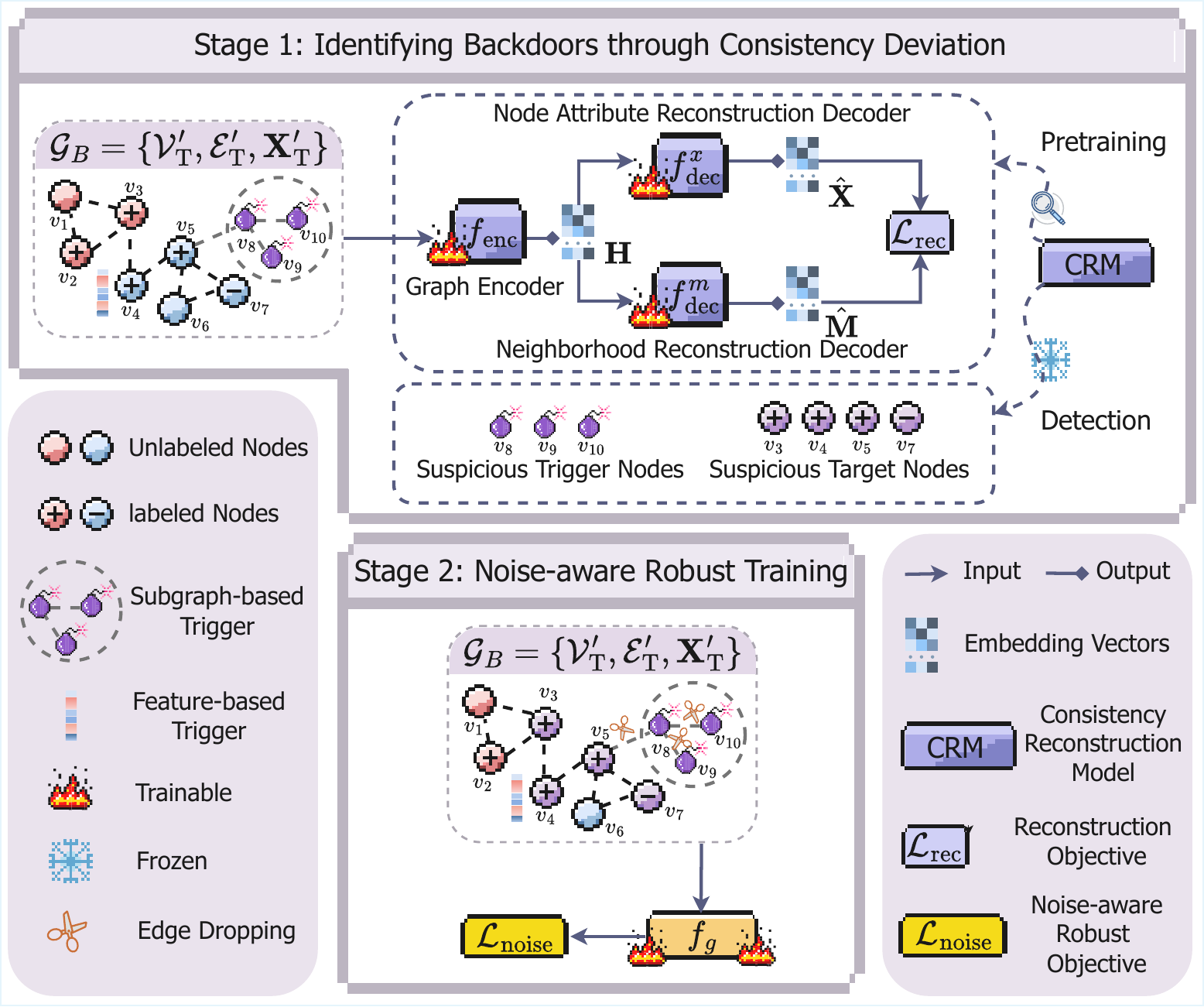} %

\caption{ Framework of CoGBD. }
\vspace{-1em}
\label{fig:framework}
\end{figure}

The above analysis reveals that backdoors exhibit pronounced discrepancies in feature consistency with their surrounding neighborhoods compared to clean nodes.
Motivated by this, we propose CoGBD (Consistency-guided Graph Backdoor Defense), a universal graph backdoor defense framework.
As depicted in Fig.~\ref{fig:framework}, CoGBD follows a two-stage pipeline.
In the first stage, a self-supervised consistency reconstruction model (CRM) is pretrained to capture dominant node--neighborhood patterns and identify both poisoned target nodes and trigger nodes via reconstruction errors.
In the second stage, to fully eliminate the influence of malicious nodes while mitigating noise introduced by detection errors, we develop a noise-aware robust training strategy for reliable model defense.

\subsection{Identifying Backdoors through Consistency Deviation}
\myparagraph{Consistency Reconstruction Pretraining.}
We characterize node--neighborhood consistency via neighborhood-informed reconstruction.
Deviations from this consistency hinder reconstruction and lead to larger reconstruction errors, facilitating the identification of nodes affected by backdoor attacks. 
To this end, we design a tri-directional reconstruction model that explicitly captures consistency among node attributes, neighborhood context, and their cross-level alignment on the backdoored graph.

\noindent\underline{\textit{Graph Encoder.}}
We use a two-layer GCN with ReLU activations as the graph encoder
$f_{\mathrm{enc}}$ to map nodes into the latent space.
Given a backdoored graph $\mathcal{G}_{\mathrm{B}}$, the hidden embeddings are computed as:
\begin{equation}
\mathbf{H} = f_{\mathrm{enc}}(\mathbf{X}'_{\mathrm{T}}, \mathbf{A}'_{\mathrm{T}})
\in \mathbb{R}^{N \times d},
\end{equation}
where $d$ denotes the hidden dimension, and $\mathbf{A}'_{\mathrm{T}}$ denotes the adjacency matrix of the backdoored graph.

\noindent\underline{\textit{Tri-directional Consistency Reconstruction Objectives.}}
As shown in Sec.~\ref{sec:rethinking}, clean nodes typically dominate a backdoored graph and preserve relatively stable node-neighborhood consistency, in contrast to target and trigger nodes.
We therefore learn normal patterns via a tri-directional reconstruction objective, such that backdoored deviations are amplified as larger reconstruction errors.
Formally, the overall reconstruction objective $\mathcal{L}_{\mathrm{rec}}$ is defined as:
\begin{equation}
\label{eq:overall-loss}
\mathcal{L}_{\mathrm{rec}}
=
\underbrace{\|\mathbf{X'}_\mathrm{T} - \hat{\mathbf{X}}\|_F^2}_{\mathcal{L}_\mathbf{node}}
+
\alpha 
\underbrace{\|\mathbf{M} - \hat{\mathbf{M}}\|_F^2}_{\mathcal{L}_\mathbf{neigh}}
+
\beta 
\underbrace{\|\mathbf{X'}_\mathrm{T} - \hat{\mathbf{M}}\|_F^2}_{\mathcal{L}_\mathbf{homo}},
\end{equation}
where $\mathbf{M}=\mathbf{A}'_{\mathrm{T}}\mathbf{X}'_{\mathrm T}$ denotes the aggregated neighborhood matrix in original graph.
Here, $\hat{\mathbf{X}} = f_{\mathrm{dec}}^{x}(\mathbf{H}) \in \mathbb{R}^{N \times F}$
and $\hat{\mathbf{M}} = f_{\mathrm{dec}}^{m}(\mathbf{H}) \in \mathbb{R}^{N \times F}$
are the reconstructed node attributes and neighborhood representations, respectively,
produced by two lightweight MLP decoders:
$f_{\mathrm{dec}}^{x}$ for node-level reconstruction and
$f_{\mathrm{dec}}^{m}$ for neighborhood-level reconstruction.
The hyperparameters $\alpha$ and $\beta$ balance the contributions of neighborhood
consistency and cross-level alignment.

Terms in Eq.~\eqref{eq:overall-loss} capture complementary aspects of graph consistency:
(i) $\mathcal{L}_{\mathrm{node}}$ preserves individual attribute regularity;
(ii) $\mathcal{L}_{\mathrm{neigh}}$ models local contextual patterns through neighborhood
reconstruction; and
(iii) $\mathcal{L}_{\mathrm{homo}}$ enforces cross-level alignment between node
attributes and their neighborhoods, corresponding to feature-based homophily.
By jointly optimizing these objectives, the reconstruction model is endowed with enhanced expressiveness to capture diverse and fine-grained node--neighborhood consistency patterns.


\myparagraph{Detection.}
The pretrained consistency reconstruction model is then leveraged to distinguish backdoors from clean nodes.
For each node $v_i$, we define its reconstruction error $e_i$ as:
\begin{equation}
\label{eq:stage2-score}
e_i
=
\|\mathbf{x}_i - \bar{\mathbf{x}}_i\|_2^2
+
\alpha \|\mathbf{m}_i - \bar{\mathbf{m}}_i\|_2^2
+
\beta \|\mathbf{x}_i - \bar{\mathbf{m}}_i\|_2^2,
\end{equation}
where $\mathbf{m}_i = [\mathbf M]_i$ is the original neighborhood representation of node $v_i$, and $\bar{\mathbf{x}}_i$ and $\bar{\mathbf{m}}_i$ denote the reconstructed node attribute and neighborhood representation, respectively.
Nodes are sorted in descending order of $e_i$. The top-$\rho$\% high-risk nodes are considered abnormal, while the remaining nodes are treated as clean.

\subsection{Noise-aware Robust Training}\label{sec:robust}

The detected backdoors can involve two types of nodes that play distinct roles in graph backdoor attacks.
Specifically, they include
(i) suspicious trigger nodes, which are typically outside the labeled training set
and serve as carriers of backdoor patterns, and
(ii) suspicious target nodes, which are labeled with the target class and inject
misleading supervision during training.

A simple and effective defense strategy is to remove edges incident to suspicious trigger nodes, which cuts off the most direct paths through which backdoor signals propagate.
However, pruning triggers alone is often insufficient.
For in-distribution attacks such as DPGBA, triggers are crafted to mimic target-class neighborhood contexts, making even a model trained on a cleaned graph still responsive to such triggers at test time.
Therefore, we further mitigate suspicious target nodes by 
minimizing the prediction confidence on their observed (typically target-class) labels, encouraging the model to counteract the impact of the trigger.
Let $\mathcal V_\mathrm{L}$ denote the labeled training nodes and
$\mathcal V_\mathrm{S}$ the suspicious target nodes.
A natural robust objective for a GNN classifier $f_g$ is therefore:
\begin{equation}
\label{eq:naive-robust}
\mathcal{L}_{\mathrm{robust}}
=
\sum_{v_i\in\mathcal V_\mathrm{L}}
\ell\!\left(f_g(v_i),y_i\right)
\;-\;
\lambda
\sum_{v_i\in\mathcal V_{\mathrm{S}}}
\ell\!\left(f_g(v_i),y_i\right),
\end{equation}
where $\ell(\cdot,\cdot)$ denotes the cross-entropy loss and
$\lambda$ controls the suppression strength applied to suspicious supervision.

\myparagraph{Noise-aware Robust Objective.}
While Eq.~\eqref{eq:naive-robust} explicitly reduces the influence of suspicious target supervision, its effectiveness still hinges on accurately distinguishing
target nodes from clean ones.
In practice, graph backdoor attacks are highly sparse and imbalanced: only a small fraction of nodes are poisoned, often accompanied by multiple trigger nodes, while the vast majority remain clean.
As a result, hard partitioning inevitably introduces false positives (i.e., clean nodes mistakenly included in $\mathcal V_S$) and false negatives (i.e., poisoned target nodes that remain in $\mathcal V_L$), which may either suppress reliable supervision or leave poisoned signals
insufficiently mitigated.
To address this issue, we incorporate reconstruction errors as a soft reliability indicator to guide a noise-aware robust training.
Formally, we define a suspiciousness score as:
\begin{equation}
s_i
=
\sigma\!\left(
-\frac{e_i - \mu_e}{\tau\,\sigma_e}
\right),
\end{equation}
where $\mu_e$ and $\sigma_e$ denote the mean and standard deviation of
reconstruction errors over the graph, $\tau$ is a temperature parameter and $\sigma(\cdot)$ is the sigmoid function.
Intuitively, larger $s_i$ values indicate stronger deviation from normal node--neighborhood consistency. 
For false positives, moderate $s_i$ should not lead to aggressive suppression, so that reliable supervision is preserved.
For false negatives, larger $s_i$ should reduce their effective influence, limiting the impact of poisoned labels even when hard separation fails.
Accordingly, we design two noise-aware node weights based on $s_i$:
\begin{equation}
w_i^{L}=(1-s_i)^{a},\qquad
w_i^{S}=s_i^{b},
\end{equation}
where $w_i^{L}$ and  $w_i^{S}$ reweight supervision on clean nodes and suspicious nodes, respectively. The exponents $a$ and $b$ are fixed and used to adjust the sharpness of the weights. 
Therefore, the final noise-aware robust objective is augmented as:
\begin{equation}
\label{eq:robust}
\mathcal{L}_{\mathrm{noise}}
=
\sum_{v_i\in\mathcal V_\mathrm{L}}
w_i^{L}\,\ell\!\left(f_g(v_i),y_i\right)
\;-\;
\lambda
\sum_{v_i\in\mathcal V_{\mathrm{S}}}
w_i^{S}\,\ell\!\left(f_g(v_i),y_i\right).
\end{equation}
The training algorithm and the corresponding complexity analysis are provided in Appendix~\ref{app:algorithem-complexity}.

\section{Experiments}
In this section, we conduct extensive experiments to answer the following research questions:
\begin{enumerate}[label=$\bullet$\ \textbf{RQ\arabic*:}]
    \item How effective is CoGBD in defending against both subgraph- and feature-based graph backdoor attacks?
    \item How well does CoGBD detect backdoors, including trigger nodes and poisoned target nodes?
    \item How do the key components contribute to the defense performance of CoGBD?
   \item How do different hyperparameters affect the ASR–ACC trade-off of CoGBD?
\end{enumerate}
\subsection{Experimental Setup}

\begin{table}[t]
\centering
\caption{Statistics of datasets.}
\label{dataset}
\small
\renewcommand{\arraystretch}{1}

\begin{tabular}{%
C{1.3cm}C{1.3cm}C{1.6cm}C{1.3cm}C{1.2cm}
}
\toprule
\textbf{Dataset}
& \textbf{$|\mathcal{V}|$}
& \textbf{$|\mathcal{E}|$}
& \textbf{\# Features}
& \textbf{\# Classes} \\
\midrule
Cora      & 2,708    & 5,429      & 1,443 & 7  \\
Pubmed    & 19,717   & 44,338     & 500   & 3  \\
Flickr    & 89,250   & 899,756    & 500   & 7  \\
OGB-arxiv & 169,343  & 1,166,243  & 128   & 40 \\
\bottomrule
\end{tabular}
\end{table}

\begin{table*}[t]
\centering
\caption{Results of backdoor defense.}
\label{tab:main}
\normalsize
\renewcommand{\arraystretch}{1.0}
\setlength{\tabcolsep}{3.0pt}

\begin{tabular}{%
C{1.1cm}|L{1.6cm}|
C{1.6cm}C{1.6cm}|
C{1.6cm}C{1.6cm}|
C{1.6cm}C{1.6cm}|
C{1.6cm}C{1.6cm}
}
\toprule
\multirow{2}{*}{\textbf{Attacks}} &
\multirow{2}{*}{\textbf{Defense}} &
\multicolumn{2}{c|}{\textbf{Cora}} &
\multicolumn{2}{c|}{\textbf{Pubmed}} &
\multicolumn{2}{c|}{\textbf{Flickr}} &
\multicolumn{2}{c}{\textbf{OGB-arxiv}} \\
\cline{3-10}
& &
\textbf{ASR(\%)}$\down$ & \textbf{ACC(\%)}$\up$ &
\textbf{ASR(\%)}$\down$ & \textbf{ACC(\%)}$\up$ &
\textbf{ASR(\%)}$\down$ & \textbf{ACC(\%)}$\up$ &
\textbf{ASR(\%)}$\down$ & \textbf{ACC(\%)}$\up$ \\
\midrule

\multirow{9}{*}{\textbf{GTA}}
& GCN        & 99.78 & 82.74 & 97.40 & 84.59 & 100.00 & 45.44 & 92.72 & 65.06 \\
& GNNGuard   & 38.97 & 77.26 & 27.22 & 79.62 & 3.60   & 44.99 & 1.05  & 65.02 \\
& RobustGCN  & 100.00& 82.52 & 100.00& 85.47 & 99.82  & 41.03 & 86.86 & 61.01 \\
& RS         & 52.99 & 75.56 & 52.63 & 84.46 & 40.83  & 41.55 & 38.71 & 61.63 \\
& ABL        & 31.37 & 81.19 & 50.68 & 84.51 & 0.00   & 41.18 & 45.19 & 64.12 \\
& Prune      & 20.59 & 82.52 & 21.88 & 84.91 & 0.00   & 41.99 & 0.09  & 66.04 \\
& OD         & 59.56 & 82.81 & 40.00 & 85.28 & 0.00   & 42.64 & 0.00  & 66.89 \\
& RIGBD      & 9.74  & 84.74 & 2.30  & 84.56 & 0.00   & 44.51 & 0.00  & 62.74 \\
\rowcolor[gray]{0.9}
\cellcolor{white} & \textbf{CoGBD} & 0.00 & 83.48 & 0.89 & 85.20 & 0.00 & 45.36 & 0.00 & 65.08 \\
\midrule

\multirow{9}{*}{\textbf{UGBA}}
& GCN        & 99.63 & 82.15 & 99.52 & 85.13 & 98.83 & 43.61 & 99.14 & 65.51 \\
& GNNGuard   & 28.41 & 77.41 & 19.95 & 80.02 & 0.00  & 43.88 & 78.65 & 66.43 \\
& RobustGCN  & 88.19 & 82.37 & 94.14 & 85.48 & 92.32 & 40.70 & 94.93 & 61.19 \\
& RS         & 49.30 & 76.96 & 50.55 & 84.40 & 27.43 & 41.27 & 27.05 & 60.62 \\
& ABL        & 55.13 & 83.11 & 28.75 & 84.72 & 0.00  & 40.89 & 79.07 & 64.21 \\
& Prune      & 97.34 & 82.00 & 100.00& 85.21 & 98.95 & 41.82 & 92.73 & 63.79 \\
& OD         & 0.00  & 83.33 & 23.45 & 84.93 & 0.00  & 41.20 & 0.04  & 65.01 \\
& RIGBD      & 13.21 & 84.00 & 11.76 & 84.51 & 0.00  & 44.87 & 0.00  & 66.33 \\
\rowcolor[gray]{0.9}
\cellcolor{white} & \textbf{CoGBD} & 0.00 & 84.30 & 1.12 & 84.33 & 0.00 & 44.36 & 2.88 & 65.42 \\
\midrule

\multirow{9}{*}{\textbf{DPGBA}}
& GCN        & 96.61 & 82.37 & 99.03 & 83.62 & 99.98 & 44.14 & 96.75 & 65.06 \\
& GNNGuard   & 89.37 & 78.07 & 96.08 & 81.73 & 92.85 & 44.57 & 96.79 & 65.07 \\
& RobustGCN  & 97.56 & 80.22 & 97.57 & 84.97 & 100.00& 41.43 & 96.58 & 65.10 \\
& RS         & 51.73 & 74.89 & 64.62 & 82.99 & 97.31 & 42.10 & 47.01 & 60.14 \\
& ABL        & 96.90 & 79.85 & 99.65 & 82.25 & 94.80 & 41.24 & 50.65 & 63.22 \\
& Prune      & 20.81 & 80.44 & 40.71 & 84.39 & 91.67 & 43.78 & 0.16  & 63.45 \\
& OD         & 94.10 & 81.78 & 99.59 & 84.37 & 99.29 & 44.53 & 96.86 & 64.70 \\
& RIGBD      & 12.40 & 82.96 & 1.61  & 83.67 & 9.81  & 43.23 & 0.02  & 65.02 \\
\rowcolor[gray]{0.9}
\cellcolor{white} & \textbf{CoGBD} & 0.00 & 84.67 & 0.00 & 85.47 & 0.00 & 43.19 & 0.06 & 65.98 \\
\midrule

\multirow{9}{*}{\textbf{SPEAR}}
& GCN        & 97.27 & 84.07 & 92.84 & 84.91 & 100.00 & 44.33 & 96.56 & 66.90 \\
& GNNGuard   & 52.55 & 79.63 & 70.12 & 81.48 & 44.86  & 44.29 & 97.51 & 67.82 \\
& RobustGCN  & 90.77 & 81.41 & 91.45 & 85.54 & 100.00 & 40.80 & 95.99 & 61.37 \\
& RS         & 85.61 & 76.22 & 87.10 & 84.68 & 96.78  & 41.23 & 83.31 & 61.19 \\
& ABL        & 97.12 & 83.93 & 93.12 & 84.58 & 20.00  & 40.90 & 95.08 & 64.41 \\
& Prune      & 99.56 & 81.85 & 95.38 & 85.00 & 95.89  & 42.01 & 98.91 & 65.65 \\
& OD         & 99.41 & 83.63 & 90.01 & 85.41 & 90.39  & 40.70 & 72.55 & 66.17 \\
& RIGBD      & 62.28 & 83.88 & 89.95 & 85.41 & 99.99  & 43.53 & 96.61 & 66.63 \\
\rowcolor[gray]{0.9}
\cellcolor{white} & \textbf{CoGBD} & 0.00 & 83.41 & 7.27 & 84.76 & 0.00 & 43.62 & 0.00 & 67.23 \\
\bottomrule
\end{tabular}
\end{table*}

\myparagraph{Datasets.}
We conduct experiments on four public real-world datasets, i.e., Cora, Pubmed~\cite{cora}, Flickr~\cite{flickr}, and OGB-arxiv~\cite{ogb},
which are widely used for inductive semi-supervised node classification.
Detailed statistics of the datasets are summarized in Table~\ref{dataset}.

\myparagraph{Attack Methods.} 
To validate the defense capability of CoGBD, we evaluate it against four
state-of-the-art graph backdoor attacks, including three subgraph-based attacks
(GTA~\cite{gta2021}, UGBA~\cite{ugba2023}, and DPGBA~\cite{dpgba2024}) and one
recent feature-based attack (SPEAR~\cite{spear2025}). Detailed descriptions are
provided in Appendix~\ref{app:attacks}.

\myparagraph{Compared Methods.}
We compare CoGBD with three competitive defense methods tailored for graph backdoor attacks, including Prune~\cite{ugba2023}, OD~\cite{dpgba2024}, and RIGBD~\cite{rigbd2024}. 
In addition, ABL~\cite{abl2021}, a popular backdoor defense in the image domain, is also included for comparison.
Moreover, we incorporate three representative robust GNN models ( RobustGCN~\cite{robustgcn2019}, GNNGuard~\cite{gnnguard2020}, and randomized smoothing (RS)~\cite{rs2021} ) to examine whether general robustness architectures can effectively mitigate graph backdoor attacks. Detailed descriptions of these defense methods are provided in Appendix~\ref{app:defense}.

\myparagraph{Evaluation Protocol.}
We follow the inductive node classification setting commonly adopted in recent graph backdoor studies~\cite{ugba2023,dpgba2024,spear2025,rigbd2024}. 
Specifically, each dataset is randomly split into two disjoint subgraphs, denoted as the training graph $\mathcal{G}_{T}$ and the unseen graph $\mathcal{G}_{U}$, with a ratio of 80:20.
The attacker is first trained on $\mathcal{G}_{T}$ to learn the trigger generation strategy. 
Subsequently, a set of target nodes $\mathcal{V}_{B}$ is selected from $\mathcal{G}_{T}$, and backdoor triggers are injected into these nodes to construct the backdoored training graph $\mathcal{G}_{B}$. 
The number of target nodes $|\mathcal{V}_{B}|$ is set to 40, 160, 160, and 565 for Cora, Pubmed, Flickr, and OGB-arxiv, respectively.
To evaluate the attack effectiveness and defense performance under the inductive setting, nodes in $\mathcal{G}_{U}$ are further divided into two equal-sized subsets. 
One subset is poisoned with the learned backdoor triggers and used to evaluate the attack success rate (ASR), while the remaining subset is kept clean and used to measure the clean accuracy (ACC). 
In addition to classification performance, we evaluate the ability of CoGBD to identify
backdoors. Specifically, we report Recall$_{\mathrm{tar}}$, the fraction of poisoned target nodes correctly detected, and Recall$_{\mathrm{tri}}$, the fraction
of trigger nodes correctly detected.

\myparagraph{Implementation Details.}
For subgraph-based attacks, following~\cite{ugba2023,dpgba2024,rigbd2024}, the trigger size is fixed to three nodes across all datasets.
For the feature-based attack SPEAR, following~\cite{spear2025}, the trigger dimension is set to $\max(0.02F, 5)$, where $F$ denotes the dimension of node features.
At the CoGBD detection stage, we set the suspicious ratio to $\rho = 3\%$ by default. During the noise-aware robust training stage, a two-layer GCN is used as the backbone classifier. The weight $\alpha$ and $\beta$ are both chosen from $[2^{-4}, 2^{-3},\cdots, 2^4 ]$. 
The temperature parameter $\tau$ is varied from 0.1 to 1.0 in increments of 0.1, and the suppression strength $\lambda$ is varied from 0.0 to 1.0 with the same step size.
Each experiment is repeated five times with different random seeds, and we report the average performance.
Full implementation details are provided in Appendix~\ref{app:implementations}.

\subsection{Performance of Defense}\label{sec:main-performance}
\myparagraph{Main Results.}
To answer \textbf{RQ1}, we compare CoGBD with representative defense methods across four datasets under both subgraph- and feature-based graph backdoor attacks.
From the results summarized in Table~\ref{tab:main} 
we draw two key observations:
(1) CoGBD consistently achieves near-zero ASR across all datasets and attack types. In contrast, RIGBD is effective only against subgraph-based attacks and fails under feature-based GBAs. For instance, under SPEAR on OGB-arxiv, RIGBD yields an ASR of 96.61\%, whereas CoGBD reduces the ASR to 0\%. 
This advantage mainly stems from our consistency reconstruction model, which accurately identifies both poisoned target nodes and trigger nodes by capturing local node--neighborhood inconsistency, thereby enabling CoGBD to effectively defend against diverse backdoor triggers and attack paradigms.
(2) 
CoGBD achieves clean accuracy comparable to, and in some cases slightly higher than, that
of the vanilla GCN.
This improvement is attributed to our noise-aware robust training objective, which
assigns node weights based on reconstruction errors for both clean and suspicious nodes,
thereby alleviating the impact of detection noise.
Overall, these results show that CoGBD achieves strong backdoor defense performance across
diverse attacks.

\myparagraph{Additional Evaluation Results.}
More performance evaluations are deferred to the appendix due to space limitations, including clean-graph performance of CoGBD (Appendix~\ref{app:clean-graph-performance}), robustness with different attack budgets (Appendix~\ref{app:attack-budget}), and transferability in different GNN backbones (Appendix~\ref{app:backbones}).

\subsection{Ability to Detect Backdoors}
\begin{table}
\centering
\caption{\label{recall}
Results for the ability to detect backdoors.}
\small
\renewcommand{\arraystretch}{1.0}
\setlength{\tabcolsep}{6pt}

\begin{tabular}{
C{0.8cm}
C{1.4cm}
C{0.6cm}
C{0.6cm}
C{1.2cm}
C{1.2cm}
}
\hline
\textbf{Attacks} 
& \textbf{Clean ACC} 
& \textbf{ASR} 
& \textbf{ACC} 
& \textbf{Recall$_{\text{tar}}$} 
& \textbf{Recall$_{\text{tri}}$} \\
\hline
GTA   & 65.17 & 0.00 & 65.08 & 100.00 & 100.00 \\
UGBA  & 65.17 & 2.88 & 65.42 & 91.27  & 100.00 \\
DPGBA & 65.17 & 0.06 & 65.98 & 85.92  & 69.27   \\
SPEAR & 65.17 & 0.00 & 67.23 & 100.00 & --     \\
\hline
\end{tabular}
\vspace{-1em}
\end{table}

To answer \textbf{RQ2}, we report in Table~\ref{recall} the recall of poisoned
target nodes ($\mathrm{Recall}_{\mathrm{tar}}$) and trigger nodes
($\mathrm{Recall}_{\mathrm{tri}}$) on OGB-arxiv, together with Clean ACC
(trained on the clean graph), ASR, and ACC.
From these results, we find that:
(1) CoGBD consistently achieves over $80\%$ recall on poisoned target nodes
under both feature-based and subgraph-based attacks, highlighting the
prevalence of feature-based homophily discrepancies induced by backdoor attacks and
thereby supporting our theoretical insights.
The effectiveness of CoGBD is particularly pronounced under feature-based GBAs, indicating that CoGBD effectively captures the most stealthy node--neighborhood inconsistencies introduced by feature-level perturbations.
(2)
CoGBD also demonstrates strong detection capability for trigger nodes under subgraph-based GBAs in most cases.
For GTA and UGBA, the recall on trigger nodes reaches $100\%$; 
For DPGBA, the recall on trigger nodes is lower, as such triggers are designed to mimic in-distribution patterns and continue to mislead the model at test time. In this case, accurately identifying poisoned target nodes is more critical for
suppressing ASR, and CoGBD effectively prioritizes target-node detection, which is sufficient to reduce the ASR to near zero despite partial trigger identification.
Overall, these results demonstrate that CoGBD reliably identifies critical backdoor-related nodes, enabling effective mitigation across diverse attack paradigms while preserving clean performance.
Additional experiments on other datasets are provided in Appendix~\ref{app:detect}.

\begin{figure*}[t]
  \centering
  \vspace{-1em}
  \begin{subfigure}[t]{0.24\textwidth}
    \centering
    \includegraphics[width=\linewidth]{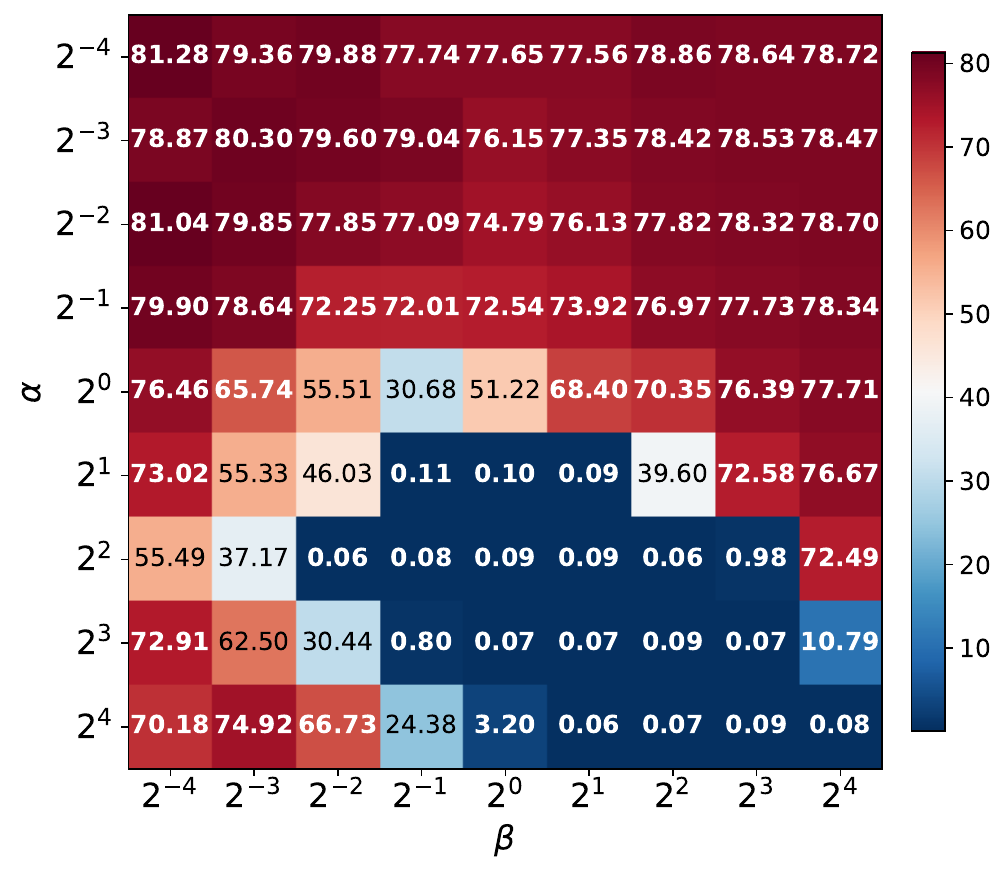}
    \caption{Heatmap of ASR.}
    \label{fig:heat_acc}
  \end{subfigure}
  \begin{subfigure}[t]{0.24\textwidth}
    \centering
    \includegraphics[width=\linewidth]{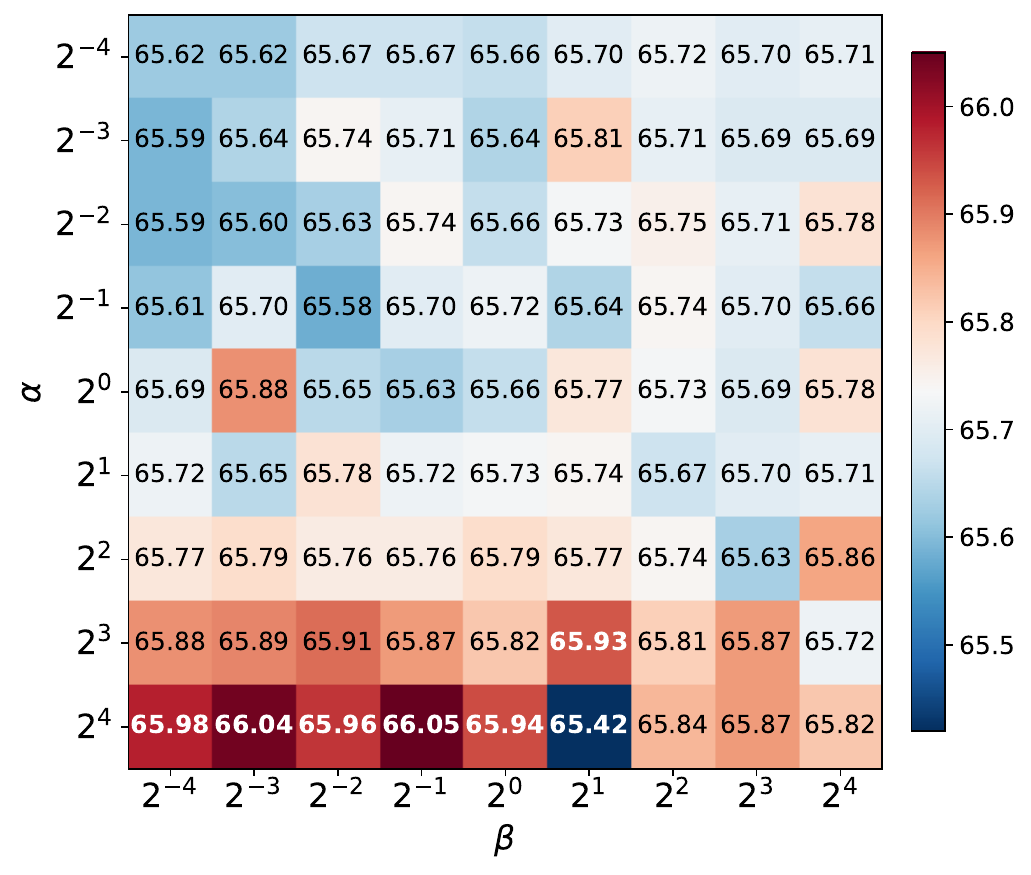}
    \caption{Heatmap of ACC.}
    \label{fig:heat_asr}
  \end{subfigure}
  \begin{subfigure}[t]{0.24\textwidth}
    \centering
    \includegraphics[width=\linewidth]{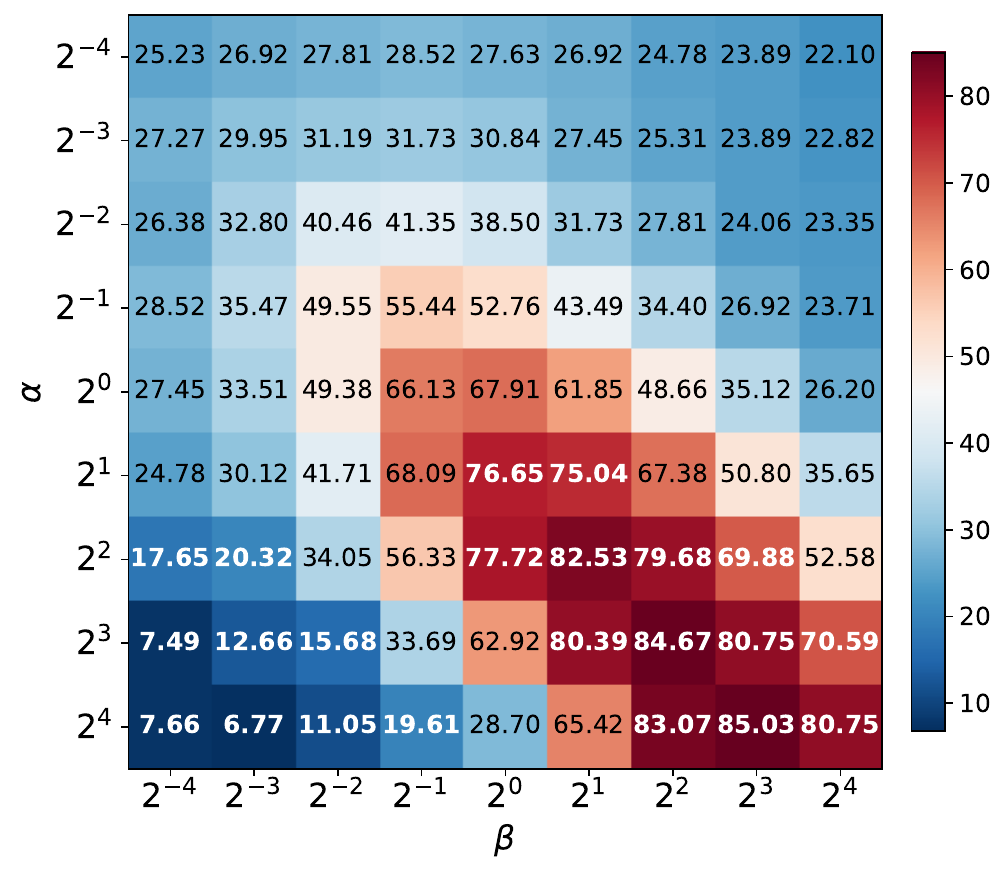}
    \caption{Heatmap of Recall$_{\text{tar}}$.}
    \label{fig:heat_recall}
  \end{subfigure}
  \begin{subfigure}[t]{0.24\textwidth}
    \centering
    \includegraphics[width=\linewidth]{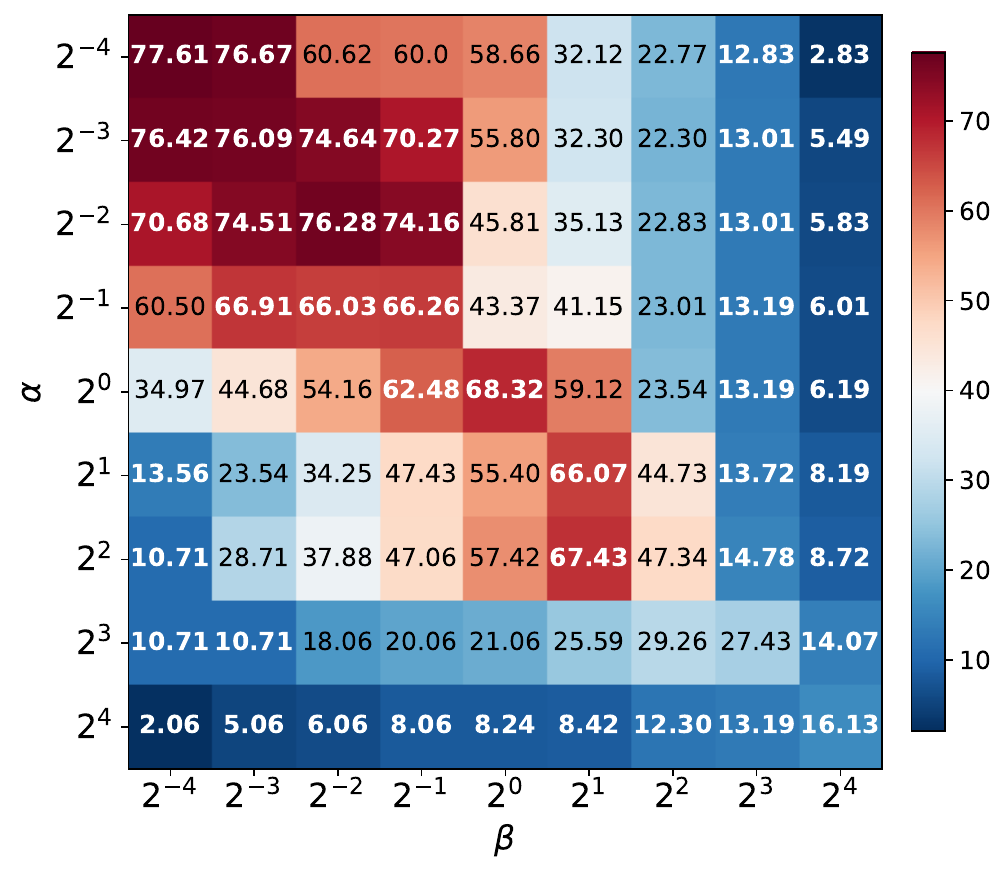}
    \caption{Heatmap of Recall$_{\text{tri}}$.}
    \label{fig:heat_precision}
  \end{subfigure}
    \vspace{-0.5em}
  \caption{Sensitivity analysis of $\alpha$ and $\beta$.}
  \label{fig:sensitivity-dpgba}
\end{figure*}
\subsection{Ablation Study}

\begin{table}
\centering
\caption{Ablation study on key components.
}
\small
\renewcommand{\arraystretch}{1.0}

\begin{tabular}{%
p{1.3cm}|
p{0.5cm}p{0.5cm}|
p{0.5cm}p{0.5cm}|
p{0.5cm}p{0.5cm}|
p{0.5cm}p{0.5cm}
}
\toprule
 & \multicolumn{2}{c|}{\textbf{GTA}}
 & \multicolumn{2}{c|}{\textbf{UGBA}}
 & \multicolumn{2}{c|}{\textbf{DPGBA}}
 & \multicolumn{2}{c}{\textbf{SPEAR}} \\
\cline{2-9}
 & ASR & ACC
 & ASR & ACC
 & ASR & ACC
 & ASR & ACC \\
\midrule
\textbf{w/o $\mathcal{L}_\mathrm{node}$}  & 0.00 & 64.42 & 55.66 & 64.63 & 0.08 & 65.81 & 0.00 & 66.42 \\
\textbf{w/o $\mathcal{L}_\mathrm{neigh}$} & 0.06 & 64.64 & 42.74 & 64.26 & 78.46 & 65.61 & 0.72 & 66.86 \\
\textbf{w/o $\mathcal{L}_\mathrm{homo}$} & 23.00 & 64.51 & 54.16 & 64.66 & 79.31 & 65.67 & 0.00 & 66.31 \\
\midrule
\textbf{w/o $\mathcal{L}_\mathrm{noise}$}  & 0.17 & 60.51 & 2.23  & 63.64 & 4.38  & 63.37 & 0.28 & 63.92 \\
\midrule
\rowcolor[gray]{0.9}
\textbf{Ours} & 0.00 & 65.08 & 2.88 & 65.42 & 0.06 & 65.98 & 0.00 & 67.23 \\
\bottomrule
\end{tabular}
\label{tab:ablation}

\end{table}

To answer \textbf{RQ3}, we compare the full CoGBD with four ablated variants:
w/o $\mathcal L_{\mathrm{node}}$, w/o $\mathcal L_{\mathrm{neigh}}$, w/o
$\mathcal L_{\mathrm{homo}}$, and w/o $\mathcal{L}_\mathrm{noise}$, which remove the node-level,
neighborhood-level, feature-homophily reconstruction signals, and the
noise-aware node weights, respectively.
We evaluate all variants on OGB-arxiv under four representative attacks and report ASR and ACC in Table~\ref{tab:ablation}.
From the table, we observe that:
(1) Removing any reconstruction signal compromises the universality of the defense, preventing the model from achieving consistently low ASR across all attack types. 
Although removing certain components remains effective for specific attacks (e.g., ``w/o $\mathcal{L}_{\text{node}}$'' on GTA, DPGBA, and SPEAR), as these attacks are more sensitive to neighborhood-level or cross-level inconsistencies, this behavior does not generalize to UGBA.
This indicates that jointly modeling node-level, neighborhood-level, and feature-based homophily reconstruction signals is essential for capturing diverse GBA-induced inconsistencies and is necessary for achieving universal defense.
(2) CoGBD consistently exhibits higher clean accuracy than ``w/o $\mathcal{L}_\mathrm{noise}$''.
For instance, CoGBD improves clean accuracy by $4.57\%$ and $3.31\%$ over ``w/o $\mathcal{L}_\mathrm{noise}$'' under GTA and SPEAR, respectively.
Unlike ``w/o $\mathcal{L}_\mathrm{noise}$'' variant, which treats all suspicious target nodes equally and is sensitive to detection noise, CoGBD incorporates noise-aware
node weights into training, thereby mitigating noisy supervision and leading to more stable training.


\subsection{Parameter Sensitivity Analysis}\label{sec:sensitivity}

\begin{figure}[t!]
  \centering
  \begin{subfigure}[b]{0.48\linewidth}
    \centering
    \includegraphics[width=\linewidth]{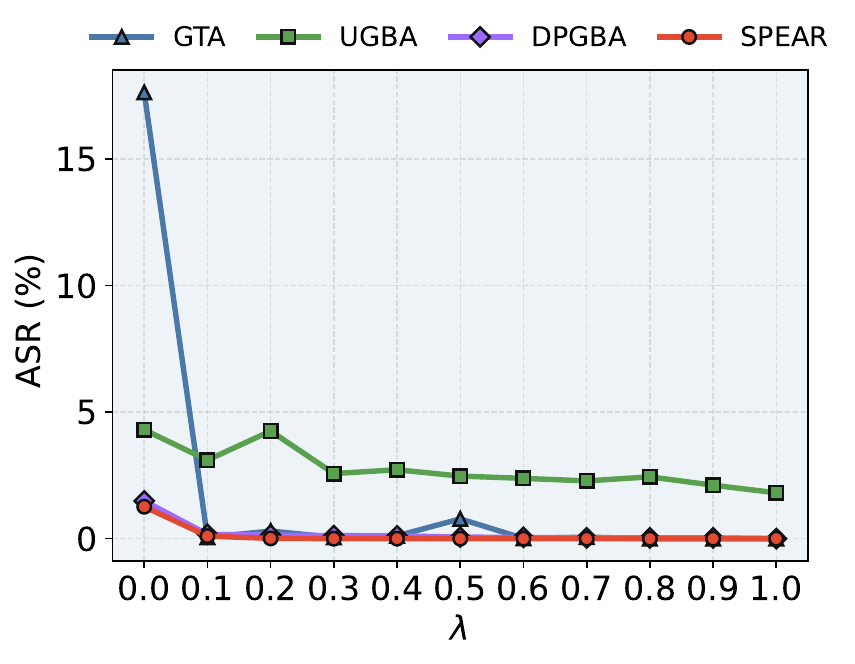}
    \caption{The results of ASR(\%).}
  \end{subfigure}
  \hfill
  \begin{subfigure}[b]{0.48\linewidth}
    \centering
    \includegraphics[width=\linewidth]{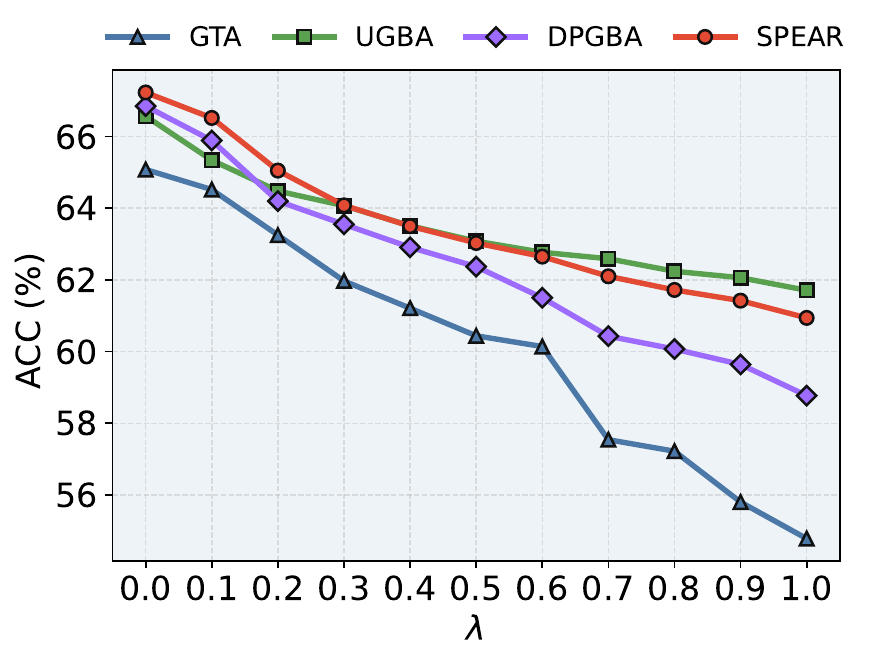}
    \caption{The results of ACC(\%).}
    
  \end{subfigure}
   \vspace{-0.5em}
  \caption{Sensitivity analysis of $\lambda$. }
  \label{fig:sensitivity-lambda}
 \vspace{-1em}
\end{figure}

To answer \textbf{RQ4}, we conduct a sensitivity analysis of CoGBD with respect to $\alpha$ and $\beta$, which balance different reconstruction signals for consistency modeling, and $\lambda$, which controls the strength of suppressing suspicious nodes during noise-aware robust training.

\myparagraph{Effect of $\alpha$ and $\beta$.}
We analyze the impact of the reconstruction weights $\alpha$ and $\beta$ by varying them from $2^{-4}$ to $2^{4}$ and evaluating CoGBD on OGB-arxiv under DPGBA,
the most stealthy subgraph-based GBA.
As shown in Fig.~\ref{fig:sensitivity-dpgba}, increasing $\alpha$ and $\beta$
steadily improves $\mathrm{Recall}_\mathrm{tar}$ while consistently reducing ASR.
In particular, when $\alpha > 0$ and $\beta > 2^{1}$, CoGBD achieves
near-zero ASR with strong clean accuracy.
This indicates that, under DPGBA, emphasizing
neighborhood-level and feature-based homophily reconstruction is more effective for identifying anomalous target nodes.
Accurate detection of poisoned target nodes directly suppresses backdoor activation, leading to lower ASR and slightly improved ACC.
Additional sensitivity results under feature-based GBA are provided in Appendix~\ref{app:sensitivity-a-b-spear}.

\myparagraph{Effect of $\lambda$.}
We study the effect of the suppression weight $\lambda$ by varying it from $0$ to $1.0$ with a step size of $0.1$, and evaluate CoGBD on OGB-arxiv under four representative attacks.
Fig.~\ref{fig:sensitivity-lambda} shows that increasing $\lambda$ consistently
reduces ASR while gradually lowering ACC, revealing a clear trade-off between
backdoor suppression and clean performance. 
This trend arises because larger values of $\lambda$ impose stronger unlearning on poisoned target nodes, effectively weakening trigger-induced patterns, but also amplify the impact of false positives, introducing training noise and degrading accuracy.
Overall, $\lambda=0.1$ provides a favorable balance in most cases, achieving low ASR while preserving high clean accuracy.

\myparagraph{More Hyperparameters Analysis.}
We analyze the impact of more hyperparameters in Appendix~\ref{app:more-senstivity} due to the space limitation.

\section{Conclusion}

In this paper, we study a novel problem of universal graph backdoor defense. To bridge subgraph-based and feature-based attacks, we introduce a feature-based homophily perspective that characterizes node–neighborhood feature consistency.
Through both theoretical analysis and empirical validation, we demonstrate that effective graph backdoor attacks consistently induce a distributional gap in feature-based homophily between clean nodes and backdoors.
Building on this insight, we propose CoGBD, a consistency-guided framework for universal graph backdoor defense.
By jointly modeling node-, neighborhood-, and cross-level consistency via multi-view reconstruction, CoGBD effectively identifies backdoors. Furthermore, a noise-aware robust training strategy mitigates the impact of detection uncertainty while eliminating trigger effects.
Extensive experiments demonstrate that our method substantially reduces the attack success rate and preserves clean accuracy across both subgraph-based and feature-based graph backdoor attacks.

\clearpage
\newpage


\clearpage
\newpage
\appendix
\section*{Appendix}
\setcounter{lemma}{0}
\renewcommand{\thelemma}{\arabic{lemma}}
\setcounter{theorem}{0}

\section{Detailed Proofs}

\noindent\textbf{Setup.}
For analytical clarity, we consider an $L$-layer linear GNN with normalized adjacency
matrix $\bar{\mathbf A}$.
Let $\mathbf H^{(l)}$ denote the node representations after the $l$-th message-passing
layer, given by
$
\mathbf H^{(l)} = \bar{\mathbf A}^{\,l}\mathbf X,
$
and let the final logits be $\mathbf Z = \mathbf H^{(l)}\mathbf W$.
For a node $v$, its $l$-layer representation is denoted by
$h_v^{(l)} \triangleq [\mathbf H^{(l)}]_v$.

\subsection{Proof of Lemma~\ref{lem:propagation}}
\label{App:propagation}
\begin{lemma}
Let $h_v^{(l)}$ and ${h'}_v^{(l)}$ denote the $l$-th layer representations
of a target node $v$ on the clean and backdoored graphs
$\mathcal{G}_{\mathrm{T}}$ and $\mathcal{G}_{\mathrm{B}}$, respectively.
We further denote
$
\pi^{(l)}_{vu} \triangleq (\bar{\mathbf A}^{\,l})_{vu},
\;
{\pi'}^{(l)}_{vu} \triangleq
\big((\bar{\mathbf A}+\Delta\bar{\mathbf A})^{\,l}\big)_{vu}.
$
Then, the representation shift
$\Delta h_v^{(l)} \triangleq {h'}_v^{(l)} - h_v^{(l)}$
can be written as:
\[
\Delta h_v^{(l)}
=
\sum_{u\in\mathcal{V}}
\big({\pi'}^{(l)}_{vu}-\pi^{(l)}_{vu}\big)\,x_u
+
\sum_{u\in\mathcal{V}}
{\pi'}^{(l)}_{vu}\,\Delta x_u ,
\]
where $x_u$ is the original feature vector of node $u$ and $\Delta x_u$ denotes the feature perturbation applied on node $u$.
\end{lemma}

\begin{proof}
On the clean training graph $\mathcal G_{\mathrm T}$, the $l$-layer representation of
node $v$ is given by:
\begin{equation}
h^{(l)}_v
=
[\bar{\mathbf A}^{\,l}\mathbf X]_v
=
\sum_{u\in\mathcal V} (\bar{\mathbf A}^{\,l})_{vu}\,x_u
=
\sum_{u\in\mathcal V} \pi^{(l)}_{vu}\,x_u.
\label{eq:clean-hv}
\end{equation}
On the backdoored graph $\mathcal G_{\mathrm B}$, the representation becomes:
\begin{align}
{h'}^{(l)}_v
&=
\big[(\bar{\mathbf A}+\Delta\bar{\mathbf A})^{\,l}
(\mathbf X+\Delta\mathbf X)\big]_v
\nonumber\\
&=
\sum_{u\in\mathcal V}
\big((\bar{\mathbf A}+\Delta\bar{\mathbf A})^{\,l}\big)_{vu}
\,(x_u+\Delta x_u)
=
\sum_{u\in\mathcal V} {\pi'}^{(l)}_{vu}\,(x_u+\Delta x_u).
\label{eq:poison-hv}
\end{align}
Subtracting~\eqref{eq:clean-hv} from~\eqref{eq:poison-hv} yields:
\begin{align}
\Delta h_v^{(l)}
\triangleq
{h'}^{(l)}_v - h^{(l)}_v
&=
\sum_{u\in\mathcal V}
\big({\pi'}^{(l)}_{vu}-\pi^{(l)}_{vu}\big)\,x_u
+
\sum_{u\in\mathcal V}
{\pi'}^{(l)}_{vu}\,\Delta x_u,
\end{align}
which completes the proof.
\end{proof}

\subsection{Proof of Lemma~\ref{lem:target}}
\label{App:shift}

\begin{lemma}
Consider a targeted graph backdoor attack on node $v$ with target class $y_t$.
The GNN classifier produces logits
$
z_{v,c} = \mathbf w_c^\top h_v^{(l)},
$
where $\mathbf w_c$ is the class-specific weight vector.
The predicted probability for class $c$ is
$
p_{v,c} = \frac{\exp(z_{v,c})}{\sum_{c'} \exp(z_{v,c'})}.
$
Let $\ell_{\mathrm{CE}}(v) = -\log p_{v,y_t}$ be the target-class cross-entropy loss.
For a small representation shift $\Delta h_v^{(l)}$ induced by the trigger,
a necessary condition for decreasing $\ell_{\mathrm{CE}}(v)$ is:
\[
\left\langle
\mathbf w_{y_t} - \bar{\mathbf w}_v,\,
\Delta h_v^{(l)}
\right\rangle > 0,
\]
where
$
\bar{\mathbf w}_v \triangleq \sum_c p_{v,c}\,\mathbf w_c
$
denotes the expected classifier weight under the current predictive distribution of node $v$.
\end{lemma}

\begin{proof}
We derive a first-order necessary condition under which a backdoor trigger can decrease
the target-class loss of node $v$.
Assume a linear classifier on top of the $l$-layer representation $h_v^{(l)}$.
The logit for class $c$ is given by
$
z_{v,c} = \mathbf w_c^\top h_v^{(l)},
$
and the predictive distribution follows the softmax
$
p_{v,c} = \exp(z_{v,c}) / \sum_{c'} \exp(z_{v,c'}).
$
The attack objective is to increase the confidence on the target class $y_t$,
which is equivalent to reducing the cross-entropy loss:
\begin{equation}
\ell_{\mathrm{CE}}(v) = -\log p_{v,y_t}.
\end{equation}
We first compute the gradient of $\ell_{\mathrm{CE}}(v)$ with respect to $h_v^{(l)}$.
By the chain rule, we have:
\begin{equation}
\label{eq:chain}
\frac{\partial \ell_{\mathrm{CE}}(v)}{\partial h_v^{(l)}}
=
\sum_{c}
\frac{\partial \ell_{\mathrm{CE}}(v)}{\partial z_{v,c}}
\frac{\partial z_{v,c}}{\partial h_v^{(l)}}.
\end{equation}
Since $\ell_{\mathrm{CE}}(v)=-\log p_{v,y_t}$, it follows that
$
\frac{\partial \ell_{\mathrm{CE}}(v)}{\partial p_{v,k}}
= -\frac{1}{p_{v,y_t}}\mathbbm{1}[k=y_t].
$
Moreover, the Jacobian of the softmax satisfies
$
\frac{\partial p_{v,k}}{\partial z_{v,c}}
= p_{v,k}\big(\mathbbm{1}[k=c]-p_{v,c}\big).
$
Combining the above relations yields:
\begin{equation}
\frac{\partial \ell_{\mathrm{CE}}(v)}{\partial z_{v,c}}
=
\sum_{k}
\frac{\partial \ell_{\mathrm{CE}}(v)}{\partial p_{v,k}}
\frac{\partial p_{v,k}}{\partial z_{v,c}}
=
p_{v,c}-\mathbbm{1}[c=y_t].
\end{equation}
Noting that $\frac{\partial z_{v,c}}{\partial h_v^{(l)}}=\mathbf w_c$,
substituting this result into Eq.~\eqref{eq:chain} gives:
\begin{equation}
\frac{\partial \ell_{\mathrm{CE}}(v)}{\partial h_v^{(l)}}
=
\sum_c \big(p_{v,c}-\mathbbm{1}[c=y_t]\big)\mathbf w_c
=
\sum_c p_{v,c}\mathbf w_c - \mathbf w_{y_t}
=
\bar{\mathbf w}_v-\mathbf w_{y_t},
\end{equation}
where $\bar{\mathbf w}_v \triangleq \sum_c p_{v,c}\mathbf w_c$.
Let $\Delta h_v^{(l)}$ denote the representation shift induced by the trigger, and define:
\begin{equation}
\Delta \ell_{\mathrm{CE}}(v)
\triangleq
\ell_{\mathrm{CE}}\!\left(h_v^{(l)}+\Delta h_v^{(l)}\right)
-
\ell_{\mathrm{CE}}\!\left(h_v^{(l)}\right).
\end{equation}
By the first-order Taylor expansion of $\ell_{\mathrm{CE}}(\cdot)$ around $h_v^{(l)}$, we obtain:
\begin{equation}
\ell_{\mathrm{CE}}\!\left(h_v^{(l)}+\Delta h_v^{(l)}\right)
=
\ell_{\mathrm{CE}}\!\left(h_v^{(l)}\right)
+
\left\langle
\nabla_{h_v^{(l)}}\ell_{\mathrm{CE}}(v),\,
\Delta h_v^{(l)}
\right\rangle
+
o\!\left(\|\Delta h_v^{(l)}\|\right),
\end{equation}
and therefore:
\begin{equation}
\Delta \ell_{\mathrm{CE}}(v)
=
\left\langle
\nabla_{h_v^{(l)}}\ell_{\mathrm{CE}}(v),\,
\Delta h_v^{(l)}
\right\rangle
+
o\!\left(\|\Delta h_v^{(l)}\|\right)
\approx
\left\langle
\bar{\mathbf w}_v-\mathbf w_{y_t},\,
\Delta h_v^{(l)}
\right\rangle.
\end{equation}
For an effective targeted attack, the target-class loss must decrease, i.e.,
$\Delta \ell_{\mathrm{CE}}(v)<0$.
Ignoring higher-order terms then yields the first-order necessary condition:
\begin{equation}
\left\langle
\mathbf w_{y_t}-\bar{\mathbf w}_v,\,
\Delta h_v^{(l)}
\right\rangle
>0,
\end{equation}
which completes the proof.
\end{proof}

\subsection{Proof of Theorem~\ref{the:homo-drop}}
\label{App:homo-drop}
\noindent\textbf{Setup.}
For clarity, we instantiate $\mathrm{AGGR}$ as mean aggregation and $\mathrm{sim}(\cdot,\cdot)$ as the inner product.
For each node $v$, define its neighborhood mean feature
$
m_v \triangleq \frac{1}{|\mathcal N(v)|}\sum_{u\in\mathcal N(v)} x_u,
$
and the feature-based homophily as
$
\mathcal H_{\mathrm{feat}}(v)=\langle x_v,\ m_v\rangle.
$

\begin{theorem}
\label{the:homo-drop}

Given a backdoored graph $\mathcal G_{\mathrm B}$ with clean node set
$\mathcal V_{\mathrm C}$ and backdoored node set $\mathcal V_{\mathrm B}$.
For any effective graph backdoor attack, the expected feature-based homophily of backdoored nodes is lower than that of clean nodes:
\[
\mathbb E_{v\sim \mathcal V_{\mathrm B}}
\!\left[\mathcal H_{\mathrm{feat}}(v)\right]
\;<\;
\mathbb E_{v\sim \mathcal V_{\mathrm C}}
\!\left[\mathcal H_{\mathrm{feat}}(v)\right].
\]
\end{theorem}
\begin{proof}
To prove that effective graph backdoor attacks necessarily introduce a systematic discrepancy in feature-based homophily between clean and backdoored nodes, 
We state two mild regularity assumptions that capture the typical behavior of clean nodes and the geometric requirement of targeted attacks.

\noindent\textbf{Key assumption (clean-majority regularity).} Clean nodes constitute the majority and follow a relatively stable local pattern, i.e.,
\begin{equation}
\mathbb E_{v\sim \mathcal V_\mathrm{C}}\!\big[\mathcal H_{\mathrm{feat}}(v)\big]
\ \ge\ \mu,
\label{eq:clean-regularity}
\end{equation}
where $\mu$ denotes the mean feature-based homophily over the whole graph.

\noindent\textbf{Key assumption (target direction is locally misaligned).}
Next, recall from Lemma~\ref{lem:target} that an effective targeted attack requires the induced
representation shift to align with the target-aware direction
\begin{equation}
\mathbf g_v \triangleq \mathbf w_{y_t}-\bar{\mathbf w}_v,
\qquad
\left\langle \mathbf g_v,\ \Delta h_v^{(l)} \right\rangle > 0.
\label{eq:target-align}
\end{equation}
For poisoned nodes, pushing representations along $\mathbf g_v$ is generally not aligned with
their original neighborhood semantics; we assume this misalignment holds in expectation:
\begin{equation}
\mathbb E_{v\sim \mathcal V_\mathrm{B}}\!\big[\langle \mathbf g_v,\ m_v\rangle\big]
\ \le\ -\gamma,
\qquad \gamma>0.
\label{eq:misalign}
\end{equation}

We now analyze how an effective attack that satisfies~\eqref{eq:target-align} affects
$\mathcal H_{\mathrm{feat}}(v)$.

\myparagraph{Case I: feature-based GBA ($\Delta\bar{\mathbf A}=\mathbf 0$).}
Feature-based attacks modify node attributes while keeping the neighborhood structure unchanged.
In this case, Lemma~\ref{lem:propagation} reduces to:
\begin{equation}
\Delta h_v^{(l)}
=
\sum_{u\in\mathcal V} \pi^{(l)}_{vu}\,\Delta x_u,
\qquad
\pi^{(l)}_{vu}=(\bar{\mathbf A}^{\,l})_{vu},
\label{eq:feat-only-prop}
\end{equation}
since $\pi'^{(l)}_{vu}=\pi^{(l)}_{vu}$.
An effective attack requires $\langle \mathbf g_v,\Delta h_v^{(l)}\rangle>0$; a canonical way
to satisfy this condition is to inject a perturbation component that is positively correlated with
$\mathbf g_v$ (e.g., concentrating $\Delta x_u$ on $u=v$ and choosing $\Delta x_v$ to have
$\langle \mathbf g_v,\Delta x_v\rangle>0$).
Under mean aggregation, the neighborhood mean $m_v$ is unchanged by feature-only attacks, and thus:
\begin{equation}
\mathcal H_{\mathrm{feat}}(v)
=
\langle x_v+\Delta x_v,\ m_v\rangle
=
\langle x_v,\ m_v\rangle
+
\langle \Delta x_v,\ m_v\rangle.
\label{eq:homo-feat-case1}
\end{equation}
For effective targeted perturbations, $\Delta x_v$ necessarily contains a component aligned with
$\mathbf g_v$, i.e., $\Delta x_v=\lambda \mathbf g_v + r_v$ with $\lambda>0$ and
$\langle \mathbf g_v, r_v\rangle=0$.
Substituting into~\eqref{eq:homo-feat-case1} gives:
\begin{equation}
\mathcal H_{\mathrm{feat}}(v)
=
\langle x_v,\ m_v\rangle
+
\lambda \langle \mathbf g_v,\ m_v\rangle
+
\langle r_v,\ m_v\rangle.
\end{equation}
Taking expectation over $v\sim\mathcal V_{\mathrm B}$ and using~\eqref{eq:misalign} yields:
\begin{equation}
\mathbb E_{v\sim\mathcal V_{\mathrm B}}\!\big[\mathcal H_{\mathrm{feat}}(v)\big]
\ \le\
\mathbb E_{v\sim\mathcal V_{\mathrm B}}\!\big[\langle x_v,\ m_v\rangle\big]
-\lambda\gamma
+
\mathbb E_{v\sim\mathcal V_{\mathrm B}}\!\big[\langle r_v,\ m_v\rangle\big].
\label{eq:case1-bound}
\end{equation}
Since $r_v$ does not contribute to satisfying the attack condition~\eqref{eq:target-align}, a rational adversary that prioritizes attack effectiveness has no incentive to choose $r_v$ to systematically increase $\langle r_v,m_v\rangle$. Accordingly, in the worst case we assume $\mathbb E[\langle r_v,m_v\rangle]\le 0$, which implies a strict drop by $\lambda\gamma$ in expectation.

\myparagraph{Case II: subgraph-based GBA ($\Delta\mathbf X=\mathbf 0$).}
Subgraph-based attacks alter the neighborhood structure of $v$ while keeping $x_v$ unchanged.
In this case, Lemma~\ref{lem:propagation} reduces to the structure-driven term:
\begin{equation}
\Delta h_v^{(l)}
=
\sum_{u\in\mathcal V}
\big(\pi'^{(l)}_{vu}-\pi^{(l)}_{vu}\big)\,x_u,
\qquad
\pi'^{(l)}_{vu}=\big((\bar{\mathbf A}+\Delta\bar{\mathbf A})^{\,l}\big)_{vu},
\label{eq:struct-only-prop}
\end{equation}
which shows that subgraph-based attacks act by reweighting the aggregation toward selected nodes.
Let $m_v$ denote the post-attack neighborhood mean under mean aggregation. Since the attack is
effective, it must satisfy~\eqref{eq:target-align}, which in this case is achieved by shifting
aggregation weights so that $\Delta h_v^{(l)}$ has positive correlation with $\mathbf g_v$.
Such reweighting inevitably drifts the neighborhood mean away from the original neighborhood
semantics of $v$ and decreases $\langle x_v,m_v\rangle$ in expectation; we assume that this effect
holds with a non-trivial margin:
\begin{equation}
\mathbb E_{v\sim\mathcal V_{\mathrm B}}\!\big[\langle x_v,\ m_v\rangle\big]
\ \le\
\mathbb E_{v\sim\mathcal V_{\mathrm B}}\!\big[\langle x_v,\ m_v^{\mathrm{clean}}\rangle\big]
-\gamma',
\qquad \gamma'>0,
\label{eq:case2-assump}
\end{equation}
where $m_v^{\mathrm{clean}}$ denotes the neighborhood mean before injecting trigger edges/subgraphs.
Therefore, subgraph-based attacks also induce a strict decrease in feature-based homophily on
poisoned nodes in expectation.

Combining both cases, effective GBAs reduce $\mathbb E_{v\sim\mathcal V_{\mathrm B}}[\mathcal H_{\mathrm{feat}}(v)]$
by a non-trivial margin, whereas clean nodes preserve a stable homophily level
as in~\eqref{eq:clean-regularity}. Consequently,
\[
\mathbb E_{v\sim \mathcal V_{\mathrm B}}
\!\left[\mathcal H_{\mathrm{feat}}(v)\right]
\;<\;
\mathbb E_{v\sim \mathcal V_{\mathrm C}}
\!\left[\mathcal H_{\mathrm{feat}}(v)\right],
\]
which proves Theorem~\ref{the:homo-drop}.
\end{proof}

\subsection{Feature-based Homophily on Trigger Nodes}
\label{app:trigger-homo}

In this section, we further discuss why trigger nodes also tend to exhibit a feature-based homophily gap on the backdoored graph.
Our goal is to explain a common phenomenon observed across diverse subgraph-based GBAs:
\begin{equation}
\label{eq:trigger-gap-goal}
\mathbb{E}_{v\sim \mathcal V_{\mathrm{tri}}}\!\left[\mathcal H_{\mathrm{feat}}(v)\right]
\;<\;
\mathbb{E}_{v\sim \mathcal V_{\mathrm C}}\!\left[\mathcal H_{\mathrm{feat}}(v)\right],
\end{equation}
where $\mathcal{V}_{\mathrm{tri}}$ denotes the trigger-node set.
We first categorize subgraph-based trigger generators according to whether they impose
an explicit feature-similarity constraint on trigger nodes.
Specifically, we consider:
(i) unconstrained triggers, such as GTA~\cite{gta2021} and DPGBA~\cite{dpgba2024}, where trigger nodes are generated primarily for attack activation
without explicitly enforcing consistency between a trigger node and its post-attachment neighborhood; and
(ii)similarity-constrained triggers, such as UGBA~\cite{ugba2023}, where trigger nodes are encouraged to resemble the attached
poisoned target node in feature space.
We show that unconstrained generators typically induce a pronounced homophily gap on trigger nodes,
while similarity constraints can narrow but generally do not eliminate the gap due to the mixed neighborhood
context after attachment.

\myparagraph{Setup.}
Following Sec.~\ref{App:homo-drop}, we instantiate $\mathrm{AGGR}$ as mean aggregation and
$\mathrm{sim}(\cdot,\cdot)$ as the inner product.
For each node $v$, define the neighborhood mean
$
m_v \triangleq \frac{1}{|\mathcal N(v)|}\sum_{u\in\mathcal N(v)} x_u
$
and $\mathcal H_{\mathrm{feat}}(v)=\langle x_v,m_v\rangle$.
For a trigger node $v_t\in\mathcal V_{\mathrm{tri}}$, its post-attachment neighborhood typically contains
a mixture of (i) poisoned targets that the trigger connects to and (ii) other trigger nodes inside the injected
subgraph.
Accordingly, we decompose the post-attachment neighborhood mean as:
\begin{equation}
\label{eq:mt-decompose}
m_{v_t}
=
\alpha_t\,\bar x_{\mathrm{tar}}(v_t)
+
\alpha_r\,\bar x_{\mathrm{tri}}(v_t)
+
\alpha_o\,\bar x_{\mathrm{oth}}(v_t),
\qquad
\alpha_t+\alpha_r+\alpha_o=1,
\end{equation}
where $\bar x_{\mathrm{tar}}(v_t)$, $\bar x_{\mathrm{tri}}(v_t)$, and $\bar x_{\mathrm{oth}}(v_t)$ are the mean
features of target neighbors, trigger neighbors, and other neighbors of $v_t$, respectively, and
$\alpha_{\{\cdot\}}$ are the corresponding neighborhood proportions.
Then
\begin{equation}
\label{eq:hvt-expand}
\mathcal H_{\mathrm{feat}}(v_t)
=
\alpha_t\langle x_{v_t},\bar x_{\mathrm{tar}}(v_t)\rangle
+
\alpha_r\langle x_{v_t},\bar x_{\mathrm{tri}}(v_t)\rangle
+
\alpha_o\langle x_{v_t},\bar x_{\mathrm{oth}}(v_t)\rangle .
\end{equation}
Eq.~\eqref{eq:hvt-expand} highlights that improving the similarity between a trigger node and the attached
poisoned target only directly affects the first term, while the overall homophily also depends on how
$x_{v_t}$ aligns with the remaining mixed neighborhood context.

\myparagraph{Unconstrained trigger generators.}
For unconstrained triggers, the trigger generator does not explicitly enforce that each trigger node remains consistent with its post-attachment neighborhood mean $m_{v_t}$.
In such cases, the neighborhood of a trigger node is often dominated by other trigger nodes (i.e., $\alpha_r$ is large) and only weakly anchored to a small set of attached poisoned targets (small $\alpha_t$).
Since the injected subgraph is engineered for attack activation rather than local semantic regularity, trigger nodes are not explicitly optimized to align with their surrounding trigger-induced context.
Accordingly, we assume that the average feature consistency among trigger nodes inside the injected subgraph is bounded by a constant:
\begin{equation}
\label{eq:unconstrained-mismatch}
\mathbb E_{v_t\sim\mathcal V_{\mathrm{tri}}}
\!\big[\langle x_{v_t},\bar x_{\mathrm{tri}}(v_t)\rangle\big]
\ \le\ \rho,
\end{equation}
where $\rho$ is generally lower than the homophily level observed on clean neighborhoods.
Combined with the mixed neighborhood composition in Eq.~\eqref{eq:hvt-expand}, this bound implies that the overall feature-based homophily $\mathcal H_{\mathrm{feat}}(v_t)$ of trigger nodes is reduced in expectation, leading to a pronounced homophily gap as in Eq.~\eqref{eq:trigger-gap-goal}.


\myparagraph{Similarity-constrained trigger generators.}
Some attacks additionally enforce a feature-similarity constraint between each trigger node $v_t$ and its attached target $v_p$, which can be written as:
\begin{equation}
\label{eq:sim-constraint-abstract}
\mathcal L_{\mathrm{sim}}
=
\sum_{(v_t,v_p)\in\mathcal E_{\mathrm{attach}}}
\ell_{\mathrm{sim}}\!\left(\langle x_{v_t},x_{v_p}\rangle\right),
\end{equation}
where $\ell_{\mathrm{sim}}(\cdot)$ is a monotonically decreasing penalty (e.g., a hinge loss).
Minimizing $\mathcal L_{\mathrm{sim}}$ directly increases $\langle x_{v_t},x_{v_p}\rangle$ and thus improves the first term in Eq.~\eqref{eq:hvt-expand}.
This effect explains why similarity-constrained triggers are typically more ``stealthy'' at the attachment point and can narrow the trigger-node homophily gap.

However, Eq.~\eqref{eq:sim-constraint-abstract} does not directly control the remaining terms in Eq.~\eqref{eq:hvt-expand}.
In particular, when $\alpha_t$ is small (i.e., the neighborhood of $v_t$ is dominated by other trigger nodes), improving $\langle x_{v_t},x_{v_p}\rangle$ can only increase $\mathcal H_{\mathrm{feat}}(v_t)$ proportionally to $\alpha_t$.
Under feature normalization (so that $|\langle x_{v_t},\cdot\rangle|\le 1$), we have the crude bound:
\begin{equation}
\label{eq:alpha-limit}
\mathcal H_{\mathrm{feat}}(v_t)
\le
\alpha_t
+
(1-\alpha_t)\cdot
\max\!\Big\{
\langle x_{v_t},\bar x_{\mathrm{tri}}(v_t)\rangle,\ 
\langle x_{v_t},\bar x_{\mathrm{oth}}(v_t)\rangle
\Big\},
\end{equation}
which makes explicit that the similarity constraint can narrow but may not eliminate the gap unless the generator also enforces
node--neighborhood consistency with respect to the entire mixed context (i.e., all terms in Eq.~\eqref{eq:hvt-expand}).
This reveals an inherent trade-off: increasing stealthiness by making trigger nodes similar to attached targets does not guarantee high homophily once the trigger is embedded into a mixed trigger--target neighborhood optimized for attack activation.

\myparagraph{Empirical validation.}
Table~\ref{tab:label-homo} reports the average feature-based homophily of clean nodes, poisoned targets, and trigger nodes computed on the backdoored graph.
Across datasets, we observe that trigger nodes produced by unconstrained generators (i.e., GTA~\cite{gta2021} and DPGBA~\cite{dpgba2024}) exhibit a clear drop in $\mathcal H_{\mathrm{feat}}$ compared with clean nodes, consistent with Eq.~\eqref{eq:trigger-gap-goal}.
Similarity-constrained generators (i.e., UGBA~\cite{ugba2023}) yield trigger nodes whose homophily is closer to that of clean nodes, indicating that the constraint in Eq.~\eqref{eq:sim-constraint-abstract} indeed narrows the gap; nevertheless, trigger-node homophily remains consistently below the clean baseline.
Overall, these results support our insight that feature-based homophily provides a stable signal for exposing trigger nodes under diverse subgraph-based GBAs. 

\section{Algorithm and Complexity}
\label{app:algorithem-complexity}

\subsection{Training Algorithm}
\label{app:algorithm}

Algorithm~\ref{alg:CoGBD} summarizes the training procedure of
CoGBD, which consists of two stages.
\textbf{Stage I: Detection of Consistency Deviation (lines 1--11).}
We first pretrain a consistency reconstruction model on the backdoored
training graph $\mathcal{G}_{\mathrm B}$ (lines 2--5).
Specifically, the encoder and decoders are optimized by minimizing the
tri-signal reconstruction objective in Eq.~\eqref{eq:overall-loss},
which captures graph consistency from node-level, neighborhood-level, and feature-based homophily.
After convergence, we compute the reconstruction error $e_i$ for each node
and select the top $\rho\%$ nodes as abnormal. 
These nodes are further partitioned into suspicious trigger nodes
(unlabeled) and suspicious target nodes (labeled).
Edges incident to suspicious trigger nodes are removed from the training graph,
while suspicious target nodes are retained for soft suppression in training.
\textbf{Stage II: Noise-aware Robust Training (lines 12--16).}
We then initialize a $L$-layer GNN classifier and train it using the
noise-aware robust objective in Eq.~\eqref{eq:robust}.
This objective down-weights supervision from suspicious target nodes while
preserving reliable signals from clean nodes, thereby mitigating the impact
of detection noise.
The final output is a GNN classifier robust to both subgraph-based and
feature-based graph backdoor attacks.



\begin{algorithm}[t]
\caption{Algorithm of CoGBD}
\label{alg:CoGBD}
\SetAlgoLined
\LinesNumbered
\KwIn{
Backdoored training graph $\mathcal{G}_{\mathrm{B}}=(\mathcal{V}'_{\mathrm{T}},\mathcal{E}'_{\mathrm{T}},\mathbf{X}'_{\mathrm{T}})$;
suspicious ratio $\rho$;
loss weights $\alpha,\beta$;
temperature parameter $\tau$;
suppression strength $\lambda$.
}
\KwOut{Backdoor-robust GNN classifier $f_g$.}

\BlankLine
\textbf{Stage I: Identifying Backdoors}\\
Initialize parameters $\theta_{\mathrm{rec}}$ of the consistency reconstruction model
(encoder $f_{\mathrm{enc}}$ and decoders $f_{\mathrm{dec}}^{x}, f_{\mathrm{dec}}^{m}$)\;
\While{not converged}{
Update $\theta_{\mathrm{rec}}$ by minimizing the overall reconstruction objective
$\mathcal{L}_{\mathrm{rec}}$ in Eq.~\eqref{eq:overall-loss} on $\mathcal{G}_{\mathrm{B}}$\;
}
\ForEach{$v_i \in \mathcal{V}'_{\mathrm{T}}$}{
Compute reconstruction error $e_i$ using Eq.~\eqref{eq:stage2-score}\;
}
Select the top $\rho\%$ nodes with largest $e_i$ as the backdoors set \;
Partition the backdoors set into:
suspicious trigger nodes (unlabeled) and
suspicious target nodes $\mathcal{V}_{\mathrm{S}}$ (labeled)\;
Remove edges incident to suspicious trigger nodes from $\mathcal{G}_{\mathrm{B}}$\;

\BlankLine
\textbf{Stage II: Noise-aware Robust Training}\\
Initialize parameters $\theta_g$ of an $l$-layer GNN classifier $f_g$\;
\While{not converged}{
Update $\theta_g$ by minimizing the noise-aware robust objective
$\mathcal{L}_{\mathrm{noise}}$ in Eq.~\eqref{eq:robust};
}
\Return Backdoor-robust GNN classifier $f_g$\;
\end{algorithm}

\subsection{Time Complexity Analysis}\label{app:complexity}
We analyze the time complexity of CoGBD and compare the pretraining process of the consistency reconstruction model with
DOMINANT~\cite{dominant2019}, a representative reconstruction-based anomaly detector for the attributed networks.
The analysis focuses on the computational cost of a single forward pass.
Specifically, $N$ and $E$ denote the number of nodes and edges, respectively. $F$ denotes the input feature dimension. $d$ denotes the hidden dimension.

\myparagraph{Consistency Reconstruction Model Pretraining.}
The model adopts an auto-encoder architecture with a two-layer GCN graph encoder and two lightweight MLP decoders with two layers. 
The encoding cost is $\mathcal{O}(E \cdot d (F + d))$.
Both decoders operate in a node-wise manner, resulting in a cost of
$\mathcal{O}(N \cdot d \cdot F)$, while the neighborhood branch additionally
requires one mean pooling operation with cost $\mathcal{O}(E \cdot d)$.
The three reconstruction objectives further introduce
$\mathcal{O}(E \cdot F + N \cdot F)$ element-wise operations.
As a result, the overall time complexity of the consistency reconstruction model per forward pass is
$
\mathcal{O}\big(
E \cdot d (F + d)
\;+\;
N \cdot d \cdot F
\big),
$
which scales linearly with the number of nodes and edges. 

\noindent\underline{\textit{Comparison with DOMINANT.}}
DOMINANT also follows an auto-encoder architecture.
However, it explicitly reconstructs the adjacency matrix and optimizes $\|\hat{\mathbf A}-\mathbf A\|_F^2$,
which incurs a dominant cost of $\mathcal{O}(N^2 \cdot d)$.
As a result, its overall time complexity per forward pass is
$\mathcal{O}(E \cdot d (F + d) + N^2 \cdot d)$,
which scales quadratically with the number of nodes.
In contrast, CoGBD avoids explicit adjacency reconstruction, making it more scalable for large graphs while retaining effective anomaly detection capability.

\myparagraph{Detection.}
Given the pretrained reconstruction model, computing reconstruction errors in Eq.~\eqref{eq:stage2-score} involves only element-wise operations over node features
and neighborhood representations, costing $\mathcal{O}(N\cdot F)$.
Selecting the top-$\rho\%$ backdoors can be implemented by sorting, costing
$\mathcal{O}(N\log N)$.
Removing edges incident to suspicious trigger nodes requires scanning edges once,
costing $\mathcal{O}(E)$.
Thus, the identification step costs
$
\label{eq:complexity-id}
\mathcal{O}\Big(NF\;+\;N\log N\;+\;E\Big).
$

\myparagraph{Noise-aware Robust Training.}
An $L$-layer GNN classifier incurs per-forward cost
$\mathcal{O}\big(L\cdot E\cdot d^2 + L\cdot N\cdot d^2\big)$ in general.
For the propagation in GCN, this is commonly dominated by sparse message passing
$\mathcal{O}(L\cdot E\cdot d^2)$.
The noise-aware reweighting in Eq.~\eqref{eq:robust} is node-wise and costs
$\mathcal{O}(N)$, which is negligible compared to message passing.
Therefore, robust training has per-forward complexity
$
\label{eq:complexity-robust}
\mathcal{O}\Big(L\cdot E\cdot d^2\Big).
$


\myparagraph{Overall Complexity.}
Combining consistency reconstruction model pretraining, detection, and noise-aware robust training, the overall complexity of CoGBD is
$
\mathcal{O}\Big(
E\cdot d(F+d)
\;+\;
N\cdot dF
\;+\;
NF
\;+\;
N\log N
\;+\;
E
\;+\;
L\cdot E\cdot d^2
\Big).
$
In typical settings where $F=\mathcal{O}(d)$, the overall cost is dominated by sparse
message passing in the reconstruction encoder and the robust GNN classifier, i.e.,
$\mathcal{O}(E\cdot d^2 + L\cdot E\cdot d^2)$, which scales linearly with $E$.

\section{Additional Experimental Setups}

\subsection{Attack Methods.}
\label{app:attacks}
We consider three representative subgraph-based GBAs (GTA~\cite{gta2021}, UGBA~\cite{ugba2023}, and DPGBA~\cite{dpgba2024}) and one
recent feature-based GBA (SPEAR~\cite{spear2025}).
A brief overview of each attack is provided below.

\begin{enumerate}[label=\arabic*.]
    \item \textbf{GTA.}
    GTA is an early graph backdoor attack that introduces adaptive, sample-specific subgraph triggers via a trigger generator optimized to minimize the backdoor attack loss.
    \item \textbf{UGBA.}
    UGBA improves attack efficiency by selecting representative target nodes through
    clustering.
    It further employs a similarity-constrained trigger generator that enforces feature similarity between trigger nodes and their attached target nodes, enhancing attack stealthiness.
    \item \textbf{DPGBA.}
    DPGBA advances subgraph-based attacks by generating in-distribution triggers via adversarial learning, making trigger nodes harder to distinguish from clean ones.
    \item \textbf{SPEAR.}
    SPEAR first identifies critical feature dimensions via a global importance-driven selection strategy, and then injects crafted feature-level triggers to maximize
    the attack success rate while preserving the original graph topology.
\end{enumerate}

\subsection{Defense Methods.}\label{app:defense}
We select three representative defense methods that are specifically designed
for graph backdoor attacks: Prune~\cite{ugba2023}, Outlier Detection (OD)~\cite{dpgba2024},
and RIGBD~\cite{rigbd2024}.
\begin{enumerate}[label=\arabic*.]
    \item \textbf{Prune.}
    Prune removes edges that connect low-similarity node pairs, based on the
    assumption that such edges are more likely to be introduced by a subgraph
    triggers.

    \item \textbf{OD.}
    OD employs a commonly used outlier detector, DOMINANT~\cite{dominant2019},
    to identify out-of-distribution nodes and removes the edges associated
    with detected anomalies.

    \item \textbf{RIGBD.}
    RIGBD first identifies poisoned target nodes by computing prediction
    variance over $K$ inference runs.
    It then estimates the target label and suppresses the confidence of
    suspicious nodes toward the predicted target class to mitigate the
    backdoor effect.
\end{enumerate}

Following~\cite{rigbd2024}, we also include a strong baseline that aims to learn
a clean model directly from backdoor-poisoned data: ABL~\cite{abl2021}.
\begin{enumerate}[label=\arabic*.]
    \item \textbf{ABL.}
    ABL is motivated by the observation that backdoor patterns are learned
    significantly faster than clean patterns during training, and that stronger
    attacks lead to faster convergence on poisoned data.
    Based on this, ABL proposes a two-stage anti-backdoor learning scheme
    that employs local gradient ascent (LGA) to first isolate backdoor samples
    at an early stage and then weaken the correlation between backdoor triggers
    and the target class, enabling the model to recover clean decision boundaries
    from poisoned data.
\end{enumerate}

We further include three representative robust GNNs: Randomized Smoothing (RS)~\cite{rs2021}, GNNGuard~\cite {gnnguard2020}, and RobustGCN~\cite{robustgcn2019}.
\begin{enumerate}[label=\arabic*.]
    \item \textbf{RS.}
    RS was originally proposed to defend against adversarial structural perturbations.
    The core idea is to construct a smoothed classifier by randomly dropping edges and aggregating predictions over multiple randomized graph instances. 
    Following~\cite{rigbd2024}, we adopt this method as a baseline and set the edge
    drop ratio to be $0.5$ to balance defense effectiveness and clean accuracy.
    \item \textbf{GNNGaurd.}
    GNNGuard adversarial structural perturbations by leveraging node similarity to reweight and prune edges. By dynamically adjusting edge importance during message passing, it suppresses the influence of adversarial connections and enables more robust propagation.
    \item \textbf{RobustGCN.}
   RobustGCN improves the robustness of GCNs against adversarial attacks by modeling node representations as Gaussian distributions rather than deterministic vectors. Adversarial perturbations are absorbed into the variances of these distributions, thereby reducing their impact on the learned representations. In addition, RobustGCN introduces a variance-based attention mechanism that assigns lower aggregation weights to uncertain neighborhoods, effectively suppressing the propagation of adversarial effects through message passing.
\end{enumerate}

\subsection{Implementation Details.}
\label{app:implementations}
\myparagraph{Computing Infrastructure.}
All methods are implemented using NumPy~1.26.3, PyTorch~2.7.1, and
PyTorch Geometric~2.6.1.
Experiments are conducted on a Linux server equipped with two Intel Xeon Silver 4314 CPUs and an NVIDIA RTX A5000 GPU with 24\,GB of memory.

\myparagraph{Attacks.}
For subgraph-based attacks, following~\cite{ugba2023,dpgba2024,rigbd2024}, the trigger size is fixed to three nodes across all datasets.
For the feature-based attack SPEAR, following~\cite{spear2025},the trigger dimension is set to
$\max(0.02F, 5)$, where $F$ denotes the node feature dimension.
A two-layer MLP is used as the trigger generator, and a two-layer GCN serves as the surrogate model.
We tune their learning rates from $\{0.0001, 0.001, 0.01, 0.05\}$ and hidden dimensions from $\{16, 32, 64, 128\}$.
All remaining hyperparameters follow the default settings of the official
implementations.

\myparagraph{Defense.}
For CoGBD, we adopt a two-layer GCN as the graph encoder, together with two separate two-layer MLPs as decoders for node attribute and neighborhood
reconstruction.
At the detection stage, the suspicious node ratio is set to $\rho = 3\%$ by default.
During noise-aware robust training, a two-layer GCN is used as the backbone
classifier.
The reconstruction weights $\alpha$ and $\beta$ are both selected from
$\{2^{-4}, 2^{-3}, \ldots, 2^{4}\}$.
The temperature parameter $\tau$ is varied from 0.1 to 1.0 in increments of 0.1, and the unlearning weight $\lambda$ is varied from 0.0 to 1.0 with the same step size.
Each experiment is repeated five times with different random seeds, and we report the average results.

\section{Additional Experimental Results}

\subsection{Results on Clean Graph}
\label{app:clean-graph-performance}
In this section, we evaluate the clean-graph performance of \textsc{CoGBD} to verify that the proposed defense does not sacrifice predictive accuracy in the absence of backdoor attacks.
To this end, we remove all backdoor triggers from the poisoned graph and train both a standard GCN and \textsc{CoGBD} on the resulting clean graph.
The results are reported in Table~\ref{tab:clean}.
As shown in the table, \textsc{CoGBD} achieves clean accuracy that is comparable to, and in some cases slightly higher than, that of the vanilla GCN.
This observation indicates that the proposed defense preserves effective supervision when no malicious perturbations are present.
The underlying reason lies in the behavior of the noise-aware robust objective in Eq.~\eqref{eq:robust}.
On clean graphs, reconstruction errors exhibit limited discriminability across nodes, which leads to nearly uniform noise-aware weights.
Consequently, the objective performs only mild reweighting instead of aggressively down-weighting or removing nodes.
This soft reweighting preserves the majority of informative training signals, thereby maintaining high clean accuracy.

\begin{table}[t]
\centering
\caption{Comparison of clean accuracy between GCN and CoGBD trained on clean graphs.}
\label{tab:clean}
\normalsize
\renewcommand{\arraystretch}{1.0}
\setlength{\tabcolsep}{6pt}

\begin{tabular}{
C{1.7cm}
C{1cm}
C{1cm}
C{1cm}
C{1.6cm}
}
\hline
\textbf{Model}
& \textbf{Cora}
& \textbf{Pubmed}
& \textbf{Flickr}
& \textbf{OGB-arxiv} \\
\hline
GCN  & 83.93 & 85.16 & 45.09 & 65.17 \\
\rowcolor[gray]{0.9}
CoGBD & 83.56  & 84.32 & 44.65 & 65.43    \\
\hline
\end{tabular}
\end{table}

\subsection{Performance with Different Attack Budgets}
\label{app:attack-budget}
In this section, we study the defense performance and detection ability under varying attack budgets, including the number of poisoned target nodes and the size of triggers.

\myparagraph{Varying the Number of Target Nodes.}
We first examine how increasing the number of poisoned target nodes affects the performance of CoGBD.
In particular, we set the number of $\mathcal{V}_\mathrm{B}$ to $\{ 113, 339, 565, 791, 1017  \}$ for OGB-arxiv under four representative GBAs.
We report the ASR before and after defense, together with ACC and the recall of target and trigger nodes.
The results are shown in Table~\ref{tab:sensitivity-number-target}.
From the table, we observe:
(1) CoGBD consistently achieves high recall for both target nodes and trigger nodes across different numbers of target nodes. In most cases, the model prioritizes the identification of trigger nodes, which can be attributed to the fact that trigger nodes in subgraph-based attacks
typically exhibit more pronounced feature-based homophily discontinuities than target nodes, making them easier to detect.
(2) CoGBD maintains low attack success rates and comparable clean accuracy throughout all settings. Even when the number of target nodes increases to 1024, where a slight decline in target-node recall can be observed, the overall defense performance remains robust.
This stability stems from the proposed noise-aware training strategy, which mitigates the impact of detection noise and prevents performance degradation under increasing attack strength.

\begin{table}[t]
\centering
\caption{
Results for defense and backdoors detection with different $|\mathcal{V}_\mathbf{B}|$.}
\label{tab:sensitivity-number-target}
\small
\renewcommand{\arraystretch}{1.0}
\setlength{\tabcolsep}{6pt}
\begin{tabular}{
C{0.7cm}
C{0.8cm}
C{2.0cm}
C{0.6cm}
C{0.9cm}
C{0.9cm}
}
\hline
\textbf{Attacks} 
& \textbf{$|\mathcal{V}_\mathrm{B}|$} 
& \textbf{ASR} 
& \textbf{ACC} 
& \textbf{Recall$_{\text{tar}}$} 
& \textbf{Recall$_{\text{tri}}$} \\
\hline
\multirow{5}{*}{\textbf{GTA}}
& 113  & 95.44 | 06.10 & 65.63 & 89.20  & 100.00 \\
& 339  & 97.97 | 00.00  & 65.20 & 100.00 & 100.00\\
& 565  & 99.96 | 00.01 & 65.20 & 100.00 & 100.00\\
& 791  & 99.97 | 00.27 & 65.02 & 100.00 & 100.00 \\
& 1,017 & 99.93 | 00.27 & 64.92 & 97.92  & 100.00\\
\midrule 
\multirow{5}{*}{\textbf{UGBA}}
& 113  & 93.42 | 03.55 & 65.26 & 80.01 & 100.00\\
& 339  & 95.49 | 03.09 & 65.29 & 78.88 & 100.00\\
& 565  & 97.55 | 00.98 & 65.37 & 65.55 & 100.00\\
& 791  & 98.26 | 00.66 & 65.49 & 66.89 & 100.00\\
& 1,017 & 98.24 | 05.41 & 64.35 & 59.00 & 100.00 \\
\midrule 
\multirow{5}{*}{\textbf{DPGBA}}
& 113  & 93.69 | 00.00 & 65.05 & 84.07 & 100.00 \\
& 339  & 95.53 | 00.00 & 65.90 & 84.42 & 100.00\\
& 565  & 96.32 | 00.00 & 65.89 & 77.01 & 69.56\\
& 791  & 97.27 | 00.00 & 65.88 & 71.08 & 72.06\\
& 1,017 & 97.29 | 00.00 & 65.00 & 73.19 & 67.65\\
\midrule 
\multirow{5}{*}{\textbf{SPEAR}}
& 113  & 96.29 | 00.06 & 66.83 & 99.07 & - \\
& 339  & 96.72 | 00.03 & 66.73 & 99.70 & - \\
& 565  & 96.78 | 00.00  & 66.66 & 94.73 & - \\
& 791  & 97.37 | 00.91 & 66.75 & 85.15 & - \\
& 1,017 & 97.86 | 03.90 & 66.78 & 76.59 & - \\
\hline
\end{tabular}
\vspace{-1em}
\end{table}

\myparagraph{Varying the Size of Triggers.}
We now conduct experiments to investigate how different trigger sizes impact the performance of CoGBD.
Specifically, we set the trigger size to $\{1, 2, 3, 4, 5 \}$ and conducted experiments on the OGB-arxiv dataset. 
We report the ASR before and after defense, together with ACC and the recall of target and trigger nodes.
The results are in Tab~\ref{tab:sensitivity-trigger-size}. 
From the table, we observe that CoGBD consistently achieves low ASR and high clean accuracy across all trigger sizes, demonstrating effective defense without
sacrificing benign performance.
As the trigger size increases, the attack success rate tends to rise; however, larger triggers are also less stealthy.
This reduced stealthiness is reflected by the increasing recall of both trigger nodes and target nodes, as larger trigger subgraphs induce more pronounced structural--semantic inconsistencies that are easier to detect.
When the trigger size is small (e.g., trigger size $=1$), the attack becomes more stealthy. In this case, subgraph-based GBAs introduce weaker semantic discontinuities on target nodes, making them harder to identify, as evidenced by the lower $\mathrm{Recall}_\mathrm{tar}$ on GTA and UGBA (e.g., 35.83\% and 51.95\%, respectively).
Despite this, CoGBD still maintains strong defense performance, which can be attributed to the proposed noise-aware mechanism that effectively mitigates the impact of detection error.

\begin{table}[t]
\centering
\caption{\label{tab:sensitivity-trigger-size}
Results for defense and backdoors detection with different trigger size.}
\small
\renewcommand{\arraystretch}{1.0}
\setlength{\tabcolsep}{6pt}
\begin{tabular}{
C{0.6cm}
C{1.5cm}
C{1.6cm}
C{0.6cm}
C{0.8cm}
C{0.8cm}
}
\hline
\textbf{Attacks} 
& \textbf{Trigger Size} 
& \textbf{ASR} 
& \textbf{ACC} 
& \textbf{Recall$_{\text{tar}}$} 
& \textbf{Recall$_{\text{tri}}$} \\
\hline
\multirow{5}{*}{\textbf{GTA}}
& 1  & 93.19 | 00.12  & 63.95 & 35.83  & 100.00 \\
& 2  & 99.89 | 00.00  & 64.02 & 100.00 & 100.00 \\
& 3  & 100.0 | 00.00  & 64.44 & 100.00 & 100.00\\
& 4  & 100.0 | 00.01  & 64.44 & 100.00 & 100.00\\
& 5  & 99.98 | 00.00  & 64.08 & 99.82  & 100.00\\
\midrule 
\multirow{5}{*}{\textbf{UGBA}}
& 1  & 83.54 | 01.80 & 64.29 & 51.95 & 100.00 \\
& 2  & 97.09 | 00.21 & 64.48 & 80.08  & 100.00\\
& 3  & 97.77 | 02.03 & 64.33 & 90.12  & 100.00\\
& 4  & 97.84 | 00.82 & 64.51 & 90.18 & 100.00 \\
& 5  & 97.73 | 05.20 & 64.26 & 90.12 & 100.00\\
\midrule 
\multirow{5}{*}{\textbf{DPGBA}}
& 1  & 92.47 | 00.00  & 64.89 & 61.95  & 100.00 \\
& 2  & 95.10 | 00.00  & 64.91 & 65.96 & 93.98\\
& 3  & 96.16 | 00.01 & 65.02 & 77.01 & 69.56\\
& 4  & 97.10 | 00.00  & 65.08 & 88.59 & 19.65 \\
& 5  & 97.65 | 00.01 & 65.02 & 93.76 & 18.23\\
\midrule 
\multirow{5}{*}{\textbf{SPEAR}}
& 1  & 96.17 | 00.01 & 66.72 & 89.92 & -\\
& 2  & 96.29 | 00.03 & 66.73 & 95.64 & -\\
& 3  & 96.61 | 00.00  & 66.68 & 95.27 & -\\
& 4  & 97.09 | 00.00  & 66.79 & 97.82 & -\\
& 5  & 97.88 | 00.01 & 66.72 & 97.45 & -\\
\hline
\end{tabular}
\vspace{-1em}
\end{table}

\subsection{More Results of the Ability to Detect Backdoors}\label{app:detect}
In this section, we provide additional results of the ability to detect backdoors on the Cora, PubMed, and Flickr datasets.
We present the $\mathrm{Recall}_\mathrm{tar}$ and $\mathrm{Recall}_\mathrm{tri}$, together with Clean ACC(which is trained on the clean graph), ASR, and ACC as reference metrics. The results are shown in Table~\ref{tab:recall-other-dataset}.
From the table, we make two main observations:
(1) CoGBD consistently achieves high recall on target nodes across different datasets and attack types. This strong target-node detection capability plays a dominant role in effectively suppressing ASR, as identifying and mitigates the influence of
target nodes, directly disrupt the backdoor attack objective.
(2) In some cases, the recall of trigger nodes is relatively lower, such as 33.33\% on Pubmed under UGBA.
This behavior can be attributed to the increased stealthiness of triggers on certain datasets, where trigger nodes exhibit weaker structural--semantic
irregularities and are therefore harder to distinguish from clean nodes. Nevertheless, the overall defense performance of CoGBD is not compromised, as the proposed noise-aware training mechanism can still effectively mitigate backdoor effects under imperfect detection.

\begin{table}[t]
\centering
\caption{\label{tab:recall-other-dataset}
Results for the ability to detect Backdoors.}
\vspace{-1em}
\small
\renewcommand{\arraystretch}{1.15}
\setlength{\tabcolsep}{6pt}
\begin{tabular}{
C{0.7cm}
C{0.8cm}
C{0.8cm}
C{0.6cm}
C{0.6cm}
C{1.0cm}
C{1.0cm}
}
\hline
\textbf{Attacks} 
& \textbf{Dataset} 
& \textbf{Clean ACC} 
& \textbf{ASR} 
& \textbf{ACC} 
& \textbf{Recall$_{\text{tar}}$} 
& \textbf{Recall$_{\text{tri}}$} \\
\hline
\multirow{3}{*}{\textbf{GTA}}
& Cora      & 83.93 & 0.0  & 83.48 & 100.00 & 33.33 \\
& Pubmed    & 85.16 & 0.89 & 85.20 & 76.00  & 100.00 \\
& Flickr    & 45.09 & 0.0  & 45.36 & 96.75  & 33.33 \\
\midrule 
\multirow{3}{*}{\textbf{UGBA}}
& Cora      & 83.93 & 0.0 & 84.30  & 80.0 & 70.83 \\
& Pubmed    & 85.16 & 1.12 & 84.33 & 88.81 & 33.33  \\
& Flickr    & 45.09 & 0.0 & 44.36  & 100.00 & 33.33 \\
\midrule 
\multirow{3}{*}{\textbf{DPGBA}}
& Cora      & 83.93 & 0.0 & 84.67 & 74.24 & 64.17 \\
& Pubmed    & 85.16 & 0.0 & 85.47 & 78.20 & 100.00 \\
& Flickr    & 45.09 & 0.0 & 43.19 & 85.62 & 100.00 \\
\midrule 
\multirow{3}{*}{\textbf{SPEAR}}
& Cora      & 83.93 & 0.0  & 83.41 & 100.00 & - \\
& Pubmed    & 85.16 & 7.27 & 84.76 & 100.00 & - \\
& Flickr    & 45.09 & 0.0  & 43.62 & 100.00 & - \\
\hline
\end{tabular}
\vspace{-1em}
\end{table}

\begin{figure*}[t]
  \centering
  \begin{subfigure}[t]{0.24\textwidth}
    \centering
    \includegraphics[width=\linewidth]{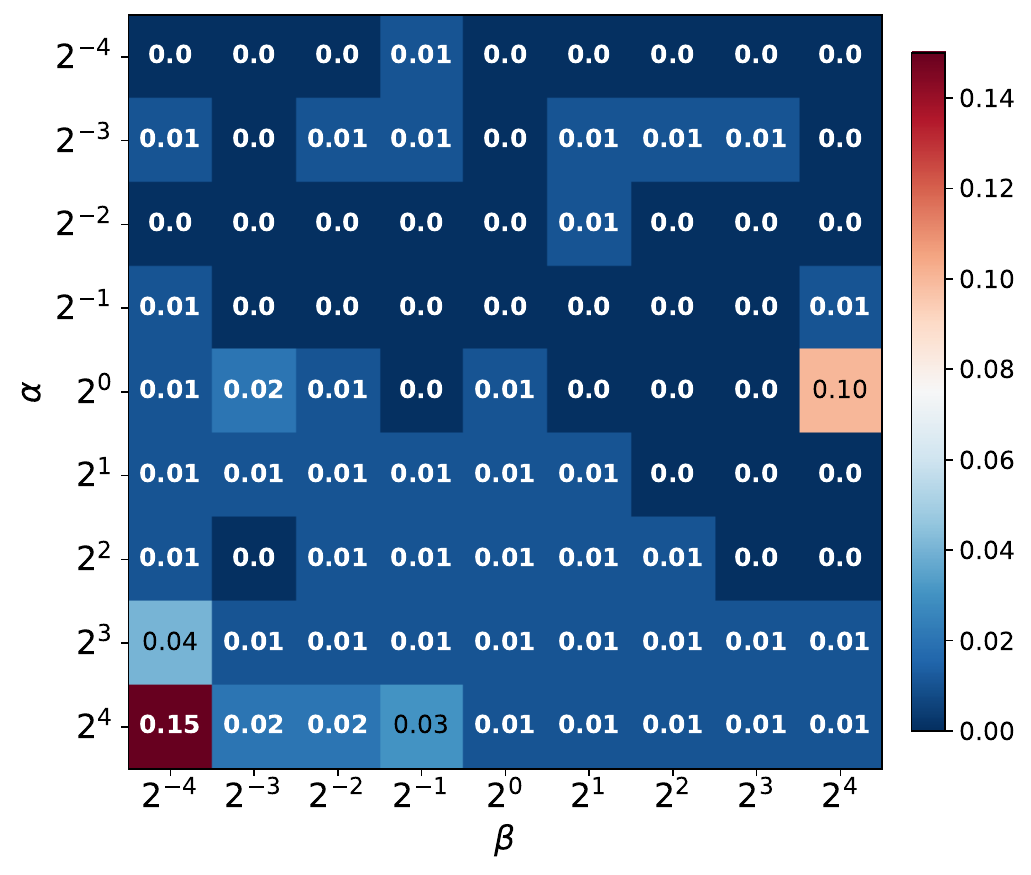}
    \caption{Heatmap of ASR.}
  \end{subfigure}
  \begin{subfigure}[t]{0.24\textwidth}
    \centering
    \includegraphics[width=\linewidth]{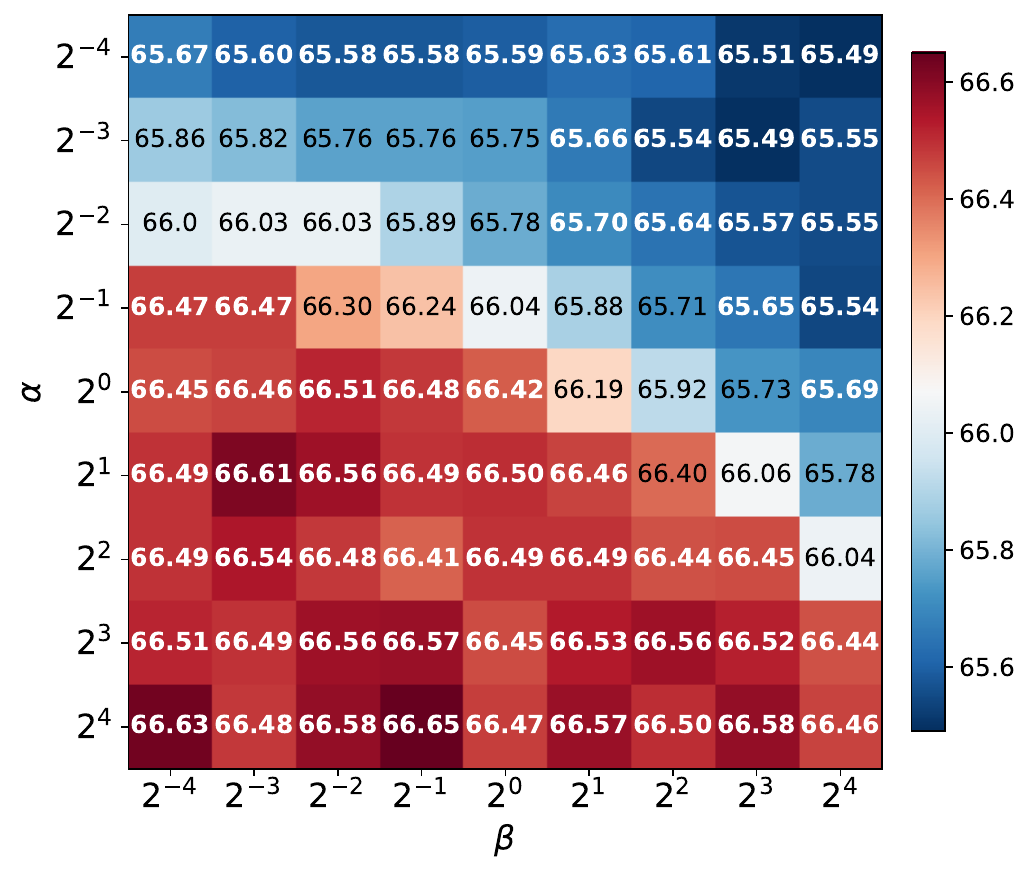}
    \caption{Heatmap of ACC.}
  \end{subfigure}
  \begin{subfigure}[t]{0.24\textwidth}
    \centering
    \includegraphics[width=\linewidth]{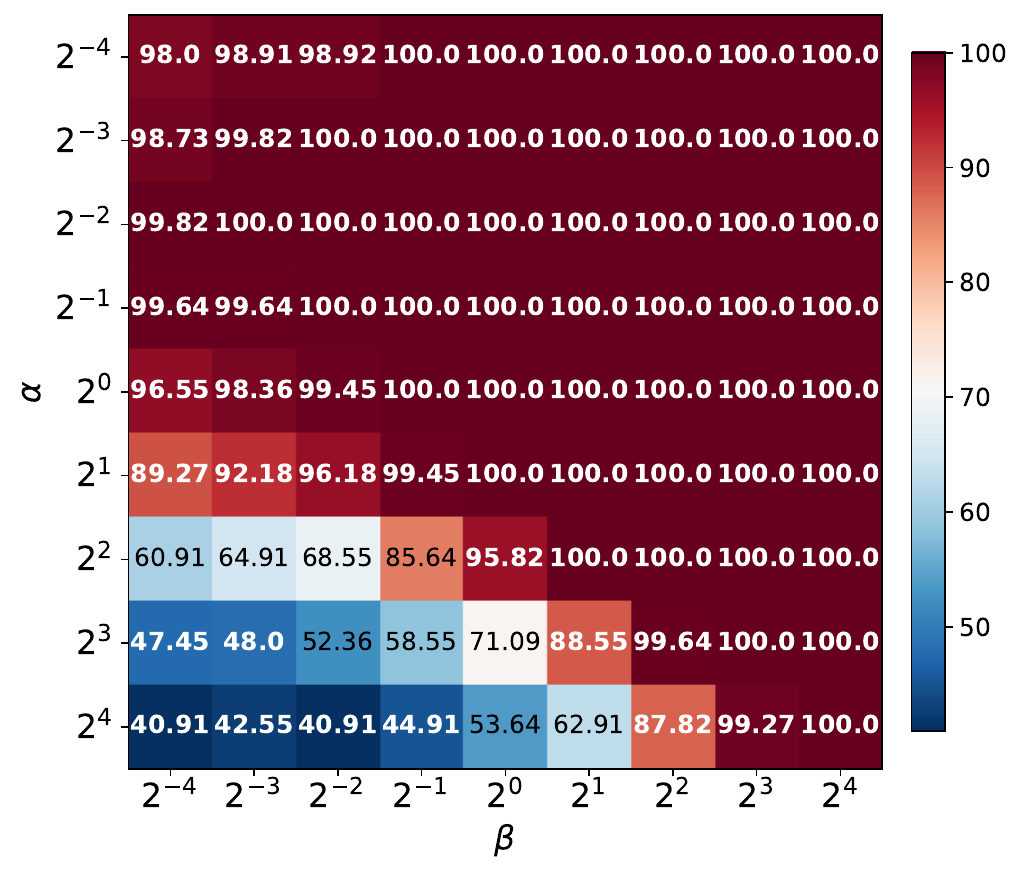}
    \caption{Heatmap of Recall.}
  \end{subfigure}
  \begin{subfigure}[t]{0.24\textwidth}
    \centering
    \includegraphics[width=\linewidth]{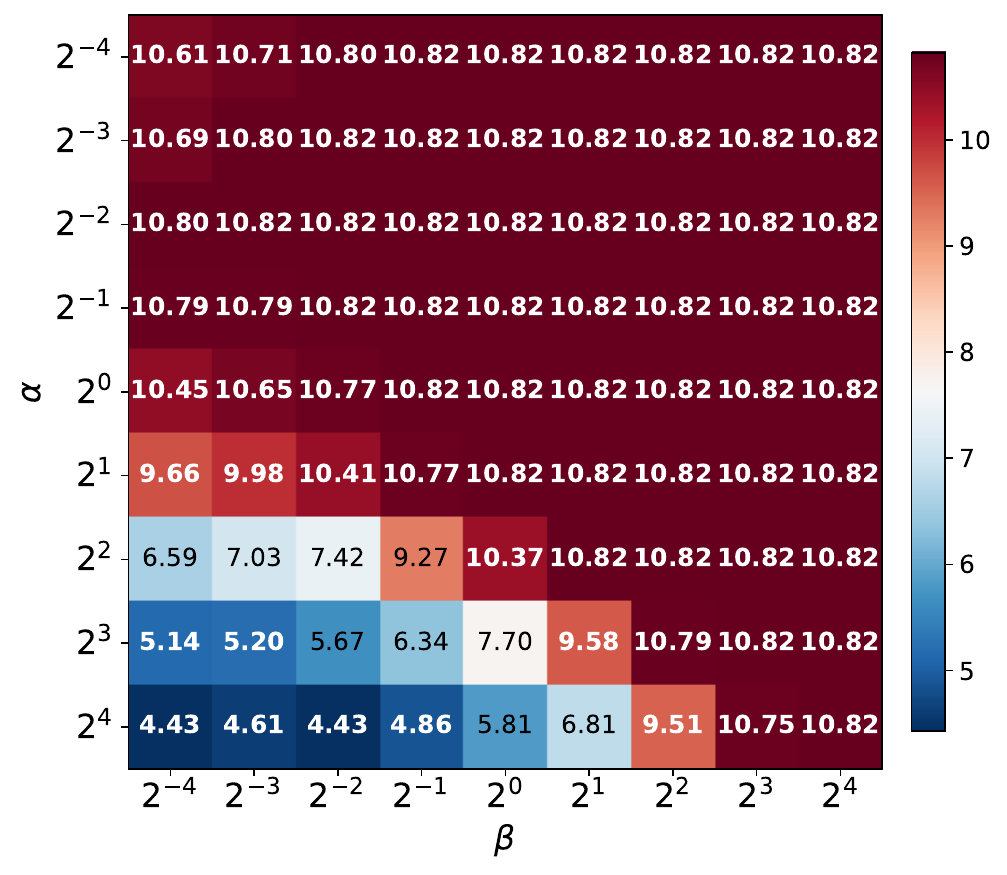}
    \caption{Heatmap of Precision.}
  \end{subfigure}
  \caption{Sensitivity analysis of weights: $\alpha$ and $\beta$.}
  \label{fig:sensitivity-spear}
\end{figure*}

\subsection{More Experiments of Parameter Sensitivity }
\label{app:more-senstivity}
\myparagraph{Parameter Sensitivity of $\alpha$ and $\beta$ under Feature-based GBAs.}
\label{app:sensitivity-a-b-spear}
Following the same setup as in Sec.~\ref{sec:sensitivity}, we vary the reconstruction
weights $\alpha$ and $\beta$ from $2^{-4}$ to $2^{4}$ and evaluate CoGBD under the feature-based SPEAR attack on OGB-arxiv.
Fig.~\ref{fig:sensitivity-spear} reports ASR, ACC, Recall$_{\text{tar}}$, and Precision.
Across a wide range of $\alpha$ and $\beta$, CoGBD consistently maintains near-zero ASR. This demonstrates the strong robustness of our CoGBD under feature-based GBAs.
This behavior stems from the nature of SPEAR, which injects triggers directly into node attributes while preserving graph structure, thereby introducing pronounced structure--feature inconsistencies around poisoned nodes.
All three reconstruction signals in CoGBD are inherently sensitive to such inconsistencies, enabling effective detection of poisoned nodes and reliable suppression of backdoor activation.
We also observe that detection precision can be relatively low due to the small fraction of poisoned nodes. Nevertheless, CoGBD preserves high clean accuracy across all settings, indicating
that the noise-aware robust training stage is tolerant to noisy detection results and prevents excessive penalization of clean nodes.
\begin{table}[t]
\centering
\caption{
Results for defense and backdoors detection with different $\rho$.}
\small
\renewcommand{\arraystretch}{1.0}
\setlength{\tabcolsep}{6pt}
\begin{tabular}{
C{0.8cm}
C{0.8cm}
C{1.0cm}
C{1.0cm}
C{1.1cm}
C{1.1cm}
}
\hline
\textbf{Attack} 
& \textbf{$\rho\%$} 
& \textbf{ASR} 
& \textbf{ACC} 
& \textbf{Recall$_{\text{tar}}$} 
& \textbf{Recall$_{\text{tri}}$} \\
\hline
\multirow{5}{*}{\textbf{GTA}}
& 1\%  & 0.00   & 65.08  & 2.85 & 100.00   \\
& 3\%  & 0.00   & 65.11  & 100.00  & 100.00   \\
& 5\%  & 0.00   & 65.18  & 100.00  & 100.00 \\
& 8\%  & 0.00   & 64.26  & 100.00  & 100.00    \\
& 10\%  & 24.37  & 63.55  & 100.00  & 100.00    \\
\midrule 
\multirow{5}{*}{\textbf{UGBA}}
& 1\%  & 0.00  & 65.28  & 0.00 & 100.00   \\
& 3\%  & 0.09 & 65.38 & 91.27 & 100.00   \\
& 5\%  & 0.87 & 64.27 & 91.88 &  100.00  \\
& 8\%  & 0.18 & 64.19 & 96.42 & 100.00   \\
& 10\%  & 0.30 & 64.13 & 97.77 & 100.00   \\
\midrule 
\multirow{5}{*}{\textbf{DPGBA}}
& 1\%  & 0.00  & 65.82 & 52.58 & 18.05   \\
& 3\%  & 0.00  & 65.97 & 80.01 & 69.56   \\
& 5\%  & 0.01 & 64.00 & 85.03 & 88.50   \\
& 8\%  & 0.01 & 64.13 & 88.95 & 97.35   \\
& 10\%  & 0.01 & 64.16 & 90.02 & 98.58   \\
\midrule 
\multirow{5}{*}{\textbf{SPEAR}}
& 1\%  & 0.00 & 66.60 & 100.00 & --   \\
& 3\%  & 0.00 & 66.35 & 100.00 & --  \\
& 5\%  & 0.00 & 66.42 & 100.00 & --   \\
& 8\%  & 0.00 & 66.49 & 100.00 & --  \\
& 10\%  & 0.00 & 66.31 & 100.00 & --   \\
\hline
\end{tabular}
\vspace{-1em}
\label{tab:sensitivity-rho}
\end{table}

\myparagraph{Parameter Sensitivity Analysis of $\rho$.}
\label{app:sensitivity-rho}
We analyze the effect of $\rho$, which controls the proportion of nodes treated as suspicious during detection.
Specifically, we vary $\rho$ over $\{10\%, 8\%, 5\%, 3\%, 1\%\}$ and report ASR, ACC, $\mathrm{Recall}_{\mathrm{tar}}$, and $\mathrm{Recall}_{\mathrm{tri}}$ on OGB-arxiv in Table~\ref{tab:sensitivity-rho}.
Our  results show that, as $\rho$ decreases, the recall on both poisoned target nodes and trigger nodes naturally declines, since fewer nodes are inspected as potential anomalies.
Notably, this reduction in recall does not necessarily impair the ability of CoGBD to suppress ASR, indicating that effective backdoor mitigation does not require exhaustive detection of all backdoors.
When $\rho$ is too large (e.g., $10\%$), excessive false positives are introduced, which injects training noise and can degrade robustness(e.g., a ASR of $24.37\%$ under GTA).
Conversely, an overly small $\rho$ (e.g., $1\%$) leads to more false negatives, resulting in missed poisoned targets (e.g., $\mathrm{Recall}_{\mathrm{tar}}=0$ under UGBA).
Overall, these results suggest that CoGBD is robust to partial detection and can effectively eliminate backdoor effects without identifying all anomalous patterns.
In our experiments, $\rho=3\%$ provides a favorable trade-off between detection coverage and robustness across most GBA settings.

\myparagraph{Parameter Sensitivity Analysis of $\tau$.}
\label{app:sensitivity-tau}
We study the effect of $\tau$, which controls the sharpness of the noise-aware node weights.
Specifically, we vary $\tau$ from $0.1$ to $1.0$ with a step size of $0.1$ and evaluate
\textsc{CoGBD} on OGB-arxiv under four representative GBAs, reporting ASR and ACC in
Fig.~\ref{fig:sensitivity-tau}.
As $\tau$ increases, ASR first decreases and then increases, while ACC exhibits the
opposite trend, indicating a clear trade-off between backdoor suppression and clean
performance.
When $\tau$ is too small, node confidence scores become insufficiently
discriminative, leading to weaker suppression of poisoned nodes and higher ASR
(e.g., $5.81\%$ ASR on UGBA at $\tau=0.1$).
In contrast, overly large $\tau$ enforces excessively sharp penalties on suspicious
nodes, amplifying the impact of false positives and introducing training noise, which
again degrades robustness (e.g., $6.07\%$ ASR on UGBA at $\tau=1.0$).
Overall, moderate values of $\tau$ provide the best balance between robustness and
accuracy.
In our experiments, $\tau \in [0.4, 0.6]$ consistently achieves low ASR while
preserving high clean accuracy across different attack settings.
\begin{figure}
  \centering
  \begin{subfigure}[b]{0.48\linewidth}
    \centering
    \includegraphics[width=\linewidth]{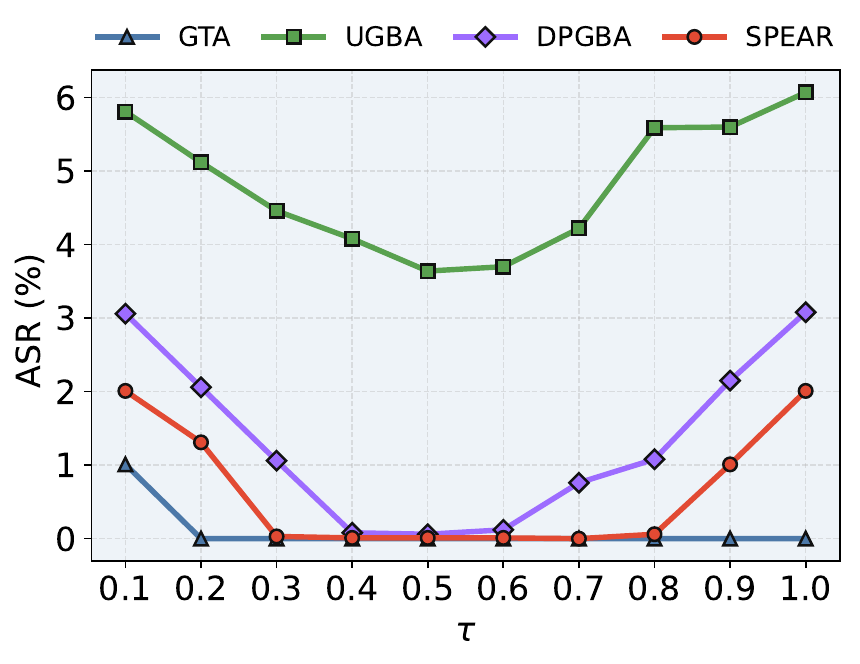}
    \caption{The results of ASR(\%).}
  \end{subfigure}
  \hfill
  \begin{subfigure}[b]{0.48\linewidth}
    \centering
    \includegraphics[width=\linewidth]{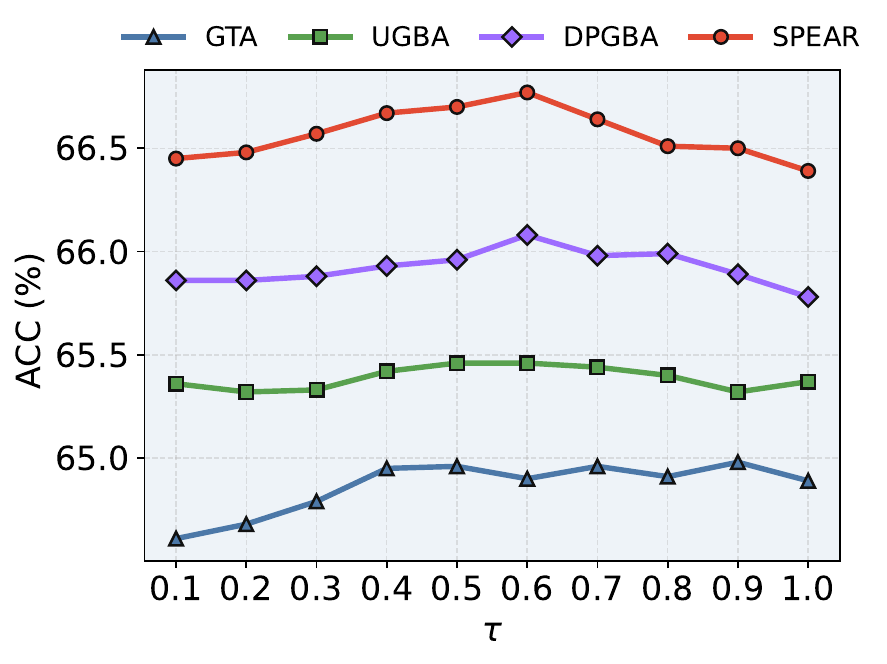}
    \caption{The results of ACC(\%).}
  \end{subfigure}
  \caption{Sensitivity analysis of $\tau$. }
  \label{fig:sensitivity-tau}
\end{figure}

\begin{table*}
\centering
\caption{Results of different GNN backbones with CoGBD.}
\label{tab:backbone}
\normalsize
\renewcommand{\arraystretch}{1.0}
\setlength{\tabcolsep}{3.0pt}

\newcommand{\HL}[1]{\cellcolor{gray!15}{#1}}

\begin{tabular}{%
C{1.2cm}|L{2.2cm}|
C{1.4cm}C{1.4cm}|
C{1.4cm}C{1.4cm}|
C{1.4cm}C{1.4cm}|
C{1.4cm}C{1.4cm}
}
\toprule
\multirow{2}{*}{\textbf{Attacks}} &
\multirow{2}{*}{\textbf{Defense}} &
\multicolumn{2}{c|}{\textbf{Cora}} &
\multicolumn{2}{c|}{\textbf{Pubmed}} &
\multicolumn{2}{c|}{\textbf{Flickr}} &
\multicolumn{2}{c}{\textbf{OGB-arxiv}} \\
\cline{3-10}
& &
\textbf{ASR(\%)}$\down$ & \textbf{ACC(\%)}$\up$ &
\textbf{ASR(\%)}$\down$ & \textbf{ACC(\%)}$\up$ &
\textbf{ASR(\%)}$\down$ & \textbf{ACC(\%)}$\up$ &
\textbf{ASR(\%)}$\down$ & \textbf{ACC(\%)}$\up$ \\
\midrule

\multirow{4}{*}{\textbf{GTA}}
& GAT              & 82.66 & 82.46 & 90.77 & 82.60 & 100.00 & 40.05 & 90.69 & 63.22 \\
& GraphSAGE        & 96.96 & 81.22 & 99.98 & 84.57 & 100.00 & 46.04 & 96.62 & 67.05 \\
& \HL{CoGBD$_\mathrm{GAT}$}        & \HL{7.49} & \HL{83.56} & \HL{1.22} & \HL{82.87} & \HL{0.00} & \HL{40.35} & \HL{0.19} & \HL{63.66} \\
& \HL{CoGBD$_\mathrm{GraphSAGE}$}  & \HL{2.29} & \HL{82.89} & \HL{4.10} & \HL{85.57} & \HL{0.00} & \HL{45.87} & \HL{0.00} & \HL{67.34} \\
\midrule

\multirow{4}{*}{\textbf{UGBA}}
& GAT              & 100.00 & 84.89 & 98.67 & 83.05 & 100.00 & 40.22 & 90.68 & 63.03 \\
& GraphSAGE        & 97.78  & 81.12 & 99.08 & 84.05 & 100.00 & 45.15 & 98.66 & 64.88 \\
& \HL{CoGBD$_\mathrm{GAT}$}        & \HL{3.91} & \HL{84.15} & \HL{5.13} & \HL{82.84} & \HL{0.00} & \HL{40.35} & \HL{0.00} & \HL{64.60} \\
& \HL{CoGBD$_\mathrm{GraphSAGE}$}  & \HL{3.70} & \HL{82.81} & \HL{5.97} & \HL{83.36} & \HL{0.00} & \HL{46.06} & \HL{1.47} & \HL{66.95} \\
\midrule

\multirow{4}{*}{\textbf{DPGBA}}
& GAT              & 91.08 & 83.33 & 95.54 & 83.56 & 100.00 & 40.35 & 91.98 & 63.02 \\
& GraphSAGE        & 90.48 & 83.35 & 93.08 & 84.01 & 100.00 & 44.35 & 90.04 & 64.56 \\
& \HL{CoGBD$_\mathrm{GAT}$}        & \HL{2.41} & \HL{83.78} & \HL{5.45} & \HL{81.60} & \HL{0.00} & \HL{40.35} & \HL{0.04} & \HL{64.11} \\
& \HL{CoGBD$_\mathrm{GraphSAGE}$}  & \HL{4.80} & \HL{82.15} & \HL{1.76} & \HL{85.40} & \HL{0.00} & \HL{43.64} & \HL{0.00} & \HL{68.02} \\
\midrule

\multirow{4}{*}{\textbf{SPEAR}}
& GAT              & 98.82 & 82.45 & 95.03 & 85.06 & 100.00 & 40.33 & 96.95 & 64.56 \\
& GraphSAGE        & 96.72 & 84.08 & 92.43 & 85.42 & 100.00 & 44.46 & 97.88 & 65.55 \\
& \HL{CoGBD$_\mathrm{GAT}$}        & \HL{0.00} & \HL{82.44} & \HL{7.40} & \HL{83.03} & \HL{0.00} & \HL{40.35} & \HL{0.37} & \HL{64.39} \\
& \HL{CoGBD$_\mathrm{GraphSAGE}$}  & \HL{3.47} & \HL{84.33} & \HL{3.83} & \HL{85.12} & \HL{0.00} & \HL{46.21} & \HL{0.02} & \HL{67.37} \\
\bottomrule
\end{tabular}
\end{table*}

\subsection{Impact of GNN backbones.}\label{app:backbones}
To evaluate the flexibility of \textsc{CoGBD} across different GNN architectures,
we replace the default GCN backbone in the noise-aware robust training stage with
alternative backbones, including GAT~\cite{gat2018} and GraphSAGE~\cite{sage2017}.
The results are reported in Table~\ref{tab:backbone}.
Across all backbones and datasets, CoGBD consistently achieves low attack
success rates while maintaining clean accuracy that is comparable to, and in some
cases slightly higher than, the corresponding vanilla GNNs.
These results indicate that the CoGBD does not rely on architecture-specific
properties and can be seamlessly integrated with diverse GNN backbones.
Overall, CoGBD exhibits strong flexibility and transferability as a general
backdoor defense framework.

\section{Additional Related Works}

\myparagraph{Graph Backdoor Attacks.}
Recent studies~\cite{SBA2021,gta2021,ugba2023,dpgba2024,spear2025} have demonstrated that
graph neural networks (GNNs) are vulnerable to backdoor attacks, posing serious security
risks in high-stakes domains where incorrect predictions may lead to severe consequences.
In a graph backdoor attack (GBA), an adversary implants trigger patterns into a small subset
of training nodes such that the trained model associates the trigger with a predefined
target label.
At test time, nodes carrying the trigger are misclassified into the target class, while
predictions on clean nodes remain largely unaffected.
In backdoor attacks, triggers are typically constructed from the basic elements of data
samples.
Analogous to pixel-level visual patterns in computer vision~\cite{image-backdoor2017,image-backdoor2019,image-backdoor2022}
and token-level triggers in natural language processing~\cite{oion2021,killer2021},
subgraphs naturally serve as triggers in graph data~\cite{SBA2021,gta2021,NFTA2023,motif2024}. 
Accordingly, most existing GBAs instantiate triggers as subgraphs by manipulating local connectivity patterns.
SBA~\cite{SBA2021} injects a fixed or sampled subgraph triggers, whereas GTA~\cite{gta2021} learns adaptive, sample-specific triggers via a generator.
UGBA~\cite{ugba2023} further enhances attack stealthiness by explicitly enforcing feature similarity between trigger nodes and the attached poisoned target nodes.
Subsequent studies reveal that many such subgraph-based triggers remain out-of-distribution and are thus detectable~\cite{dpgba2024}, motivating DPGBA~\cite{dpgba2024} to generate in-distribution triggers through adversarial optimization.
Beyond structural manipulation, feature-based GBAs inject triggers directly into node attributes while preserving the original graph topology.
SPEAR~\cite{spear2025} exemplifies this direction and poses greater challenges to existing defenses, as it introduces no explicit structural anomalies and relies purely on feature-level perturbations.
From a supervision perspective, the aforementioned methods primarily operate under the dirty-label setting, where poisoned nodes are explicitly relabeled to the target class during training.
This setting simplifies attack construction but is less stealthy due to explicit label manipulation.
More recently, several works~\cite{poster2022,cgba2025} have explored the
clean-label setting, in which training labels remain unchanged.
Instead, attackers rely on subtle trigger-induced representation shifts to cause misclassification at test time.
Moreover, recent studies~\cite{mlgba2024,eumc2024} extend GBAs to the multi-target setting, where multiple target labels are attacked simultaneously. Compared with single-target attacks, multi-target GBAs require learning more complex and potentially overlapping trigger–label associations, which increases both the attack surface and the difficulty of defense. This setting further exacerbates the challenge for existing defenses, as abnormal patterns may be distributed across different target classes rather than concentrated around a single label.


\end{document}